\documentclass[12pt,english]{article}
\newif\ifanon
\anonfalse

\usepackage[T1]{fontenc}
\usepackage[latin9]{inputenc}
\usepackage{verbatim}
\usepackage{booktabs}
\usepackage{amsmath}
\usepackage{amsthm}
\usepackage{amssymb}
\usepackage{graphicx}
\usepackage{geometry}
\usepackage{setspace}
\usepackage[authoryear]{natbib}
\makeatletter

\providecommand{\tabularnewline}{\\}

\theoremstyle{definition}

\theoremstyle{plain}
\newtheorem{prop}{\protect\propositionname}
\theoremstyle{plain}

\theoremstyle{plain}
\newtheorem{lem}{\protect\lemmaname}
\theoremstyle{plain}

\theoremstyle{plain}
\newtheorem{thm}{\protect\theoremname}
\theoremstyle{definition}
 \newtheorem{example}{\protect\examplename}
\newtheorem{assump}{\protect\assumptionname}

\usepackage{graphicx}
\usepackage[usenames,dvipsnames]{color}
\definecolor{ChadDarkBlue}{rgb}{.1,0,.2}
\definecolor{ChadBlue}{rgb}{.1,.1,.5}
\definecolor{ChadRoyal}{rgb}{.2,.2,.8}
\definecolor{ChadGreen}{rgb}{0,.4,0}    
\definecolor{ChadRed}{rgb}{.5,0,.5}  

\definecolor{DarkerGreen}{rgb}{0,.35,0}    
\definecolor{ChenBlack}{rgb}{0,0,0}    
\definecolor{ChenGreen}{rgb}{0,.30,0}    

\definecolor{ChenBlue}{rgb}{0,0,.5}
\definecolor{ChenRed}{rgb}{0.6,0,0}
\definecolor{ChenGreen}{rgb}{0,.6,0}    

\usepackage{booktabs,bigstrut,multirow,tabularx,caption,subcaption, threeparttable, rotating, lscape}
\usepackage{appendix}

\usepackage{titlesec}
\titlelabel{\thetitle. \quad}
\titleformat*{\section}{\normalfont\Large\bfseries\color{ChadBlue}}
\titleformat*{\subsection}{\normalfont\large\bfseries\color{ChadBlue}}
\titleformat*{\subsubsection}{\normalfont\normalsize\bfseries\color{ChadBlue}}
\titleformat*{\paragraph}{\normalfont\normalsize\bfseries\color{ChadBlue}}

\usepackage{hyperref} 
\hypersetup{hidelinks}
 \hypersetup{
     colorlinks=false,
     linkbordercolor=[rgb]{0,0.8,0},
     citebordercolor=[rgb]{1,1,1},
     urlbordercolor =[rgb]{1,1,1},
     }

\usepackage{pdfsync}

\usepackage{indentfirst}
\usepackage[shortcuts]{extdash}
\usepackage{fourier}

\renewcommand\[{\begin{equation}}
\renewcommand\]{\end{equation}}
\renewenvironment{align*}{\align}{\endalign}

\makeatother

\usepackage{babel}
\providecommand{\conjecturename}{Conjecture}
\providecommand{\corollaryname}{Corollary}
\providecommand{\definitionname}{Definition}
\providecommand{\examplename}{Example}
\providecommand{\lemmaname}{Lemma}
\providecommand{\propositionname}{Proposition}
\providecommand{\theoremname}{Theorem}
\providecommand{\assumptionname}{Assumption}
\providecommand{\resultname}{Result}

\begin{document}

\title{\textbf{\color{ChadBlue}An Irrelevance Theorem for Risk Aversion and Time-Varying Risk}}
\ifanon
    \date{}
\else
    \author{\large {Andrew Y. Chen}\\{\normalsize Federal Reserve Board} \and \large{Francisco Palomino}\\{\normalsize Federal Reserve Board}}
    \date{April 2026\thanks{We thank Anthony Diercks, Mohammad Jahan-Parvar, Jae Sim, Hiro Tanaka,
    Paco Vazquez-Grande, and seminar participants at the Federal Reserve
    Board, 2018 Latin American Meetings of the Econometric Society, and
    2018 Southern Economic Association Meetings for helpful comments.
    The views expressed herein are those of the authors and do not necessarily
    reflect the position of the Board of Governors of the Federal Reserve
    or the Federal Reserve System.}}
\fi
\maketitle

\thispagestyle{empty}

\begin{abstract}
\noindent We provide a theorem on the role of risk and risk attitudes in macroeconomic models that clarifies and extends the Tallarini (2000) separation result. Under (1) separation of intertemporal and risk preferences, (2) separation of drivers of first and higher moments in the model primitives, and (3) approximate linearity of constraints, risk aversion and time-varying risk are irrelevant for the elasticity of any endogenous variable with respect to state variables that don't drive variation in higher moments.  We discuss how models generate a more prominent role for risk by ``breaking'' or ``adapting'' to the assumptions in the theorem.
\end{abstract}


\vspace{5cm}

\pagebreak{}

\setcounter{page}{1}

\section{Introduction}

Models that incorporate advances in specifying risk attitudes and risk dynamics, such as exotic preferences, long run risk, or disaster risk, have deepened our understanding of asset prices in endowment economies. But in standard real business cycle (production) economies, despite the separation identified in \citet{tallarini2000risk}, these advances struggle to generate enough action from risk and simultaneously match salient macroeconomic and financial empirical regularities.\footnote{
Exotic preferences, as the ones described in \citet{backus2004exotic} and \citet{cochrane2017macro}, include recursive preferences (\citet*{epstein1989substitution}; \citet{weil1990nonexpected}) and external habit (\citet*{campbell1999force}).  Endowment economy versions of long-run risks  (\citet*{bansal2004risks}) and time-varying disaster risk (\citet{gabaix2012variable}, \citet{wachter2013can}) match numerous equity price facts, but production economy versions of long-run risks (\citet{kaltenbrunner2010long}, \citet{croce2014long}) and disaster risk (\citet{gourio2012disaster}) produce tiny equity price volatility, unless an elevated degree of financial leverage is included.} %
A much richer economic structure is needed in production economies, such as nominal rigidities in \citet{basu2017uncertainty}, or market incompleteness in \citet{DiTellaHall2022}, to give risk and risk attitudes a more prominent role in describing observed economic dynamics.%

We provide a theorem that sheds light on this phenomenon. We write down a representative agent model with \citet{epstein1989substitution} (EZ) recursive preferences and a general, abstract structure. EZ preferences separate attitudes towards intertemporal substitution and risk aversion. The remaining economic structure assumes a similar separation: transformation of resources across time is separated from transformation across future states of nature. 

We use the terms ``first moment states'' and ``higher moment states'' to describe this separation. First moment states are variables that control expected values of the model structure, while higher moment states control higher central moments. For example, in neoclassical models, capital are productivity are first moment states, while the volatility of productivity is a higher moment state. 

In the real world, first and higher moments states are intertwined---production decisions involve risk-return tradeoffs. But in standard models, perhaps for tractability, first moment and higher moment states are often separated. 

Our theorem states that, if preferences and the remaining model structure exhibit these separations, and this structure is approximately linear, then risk aversion and time-varying risk are irrelevant for the elasticity of any variable with respect to first moment states. 

The theorem implies that augmenting a stochastic growth model with exogenous shocks to volatility or risk aversion has little effect on impulse responses with respect to productivity, government policy, or expected time preference. Indeed, the impulse responses with respect to these first moment states are well approximated by solving a variation of the model with no uncertainty at all. Including nominal rigidities or other frictions does not affect this irrelevance, unless these frictions connect first and higher moment states, or generate strong nonlinearities.

The intuition is found in the Euler equation. Separation in preferences implies that the Euler equation can be decomposed into a piece that characterizes intertemporal tradeoffs and a piece that characterizes risk tradeoffs. In principle, these tradeoffs interact through the remaining structure, such as the production technology. But if the remaining structure also exhibits separation, then these direct interactions are ruled out. Still, indirect interactions could occur through nonlinearities, but in many models the structure is approximately linear. As a result, the intertemporal piece of the Euler equation completely determines elasticities with respect to first-moment states, and for these elasticities, risk tradeoffs are irrelevant. In short, the understandable desire to separate economic channels leads to a stark separation in equilibrium conditions, and the irrelevance of risk for certain dynamics.

In the special case that shocks are homoskedastic, the theorem implies that risk aversion has no effect on the impulse response for \emph{any} variable, regardless of whether the variable is macroeconomic or financial. This special case generalizes and clarifies the \citet{tallarini2000risk} separation result, which is widely understood as a separation between drivers of macroeconomic and financial dynamics (e.g., \citealt{lucas2003macroeconomic}, \citealt{cochrane2017macro}).

Instead, our theorem shows the separation is actually between means and volatilities. Indeed, though risk aversion has no effect on elasticities under homoskedasticity, we show that it has clear effects on model intercepts and steady states. A taste of this result is found in \citet{tallarini2000risk}, who shows that risk aversion determines the average equity premium but has no effect on risk-free rate volatility. A similar result is seen for bond price dynamics in \citet*{rudebusch2012bond} and stock price dynamics in \citet*{petrosky2018endogenous}. These models span a wide variety of economic structures, from New-Keynesian to labor search, illustrating the scope of our theorem.

If shocks are heteroskedastic, the theorem's implications are less stark. In this case, risk aversion is critical for elasticities with respect to higher moment states, and thus key model dynamics. 

But combined with the \citet*{barro1984time} comovement result, the theorem helps explain the difficulties of modeling asset prices in standard production economies. Barro and King show that, in the neoclassical structure, many shocks fail to generate the observed comovement of consumption, investment, and labor.%
\footnote{%
We derive this result under recursive preferences in Appendix \ref{subsec:app_barro_king}. 
} %
This result has little impact on endowment economy models, in which the comovement of consumption and investment are ignored. But in neoclassical production economies, Barro-King implies that the magnitude of higher moment shocks must be limited. Indeed, production economy versions of asset pricing models often drive business cycles with traditional productivity shocks (e.g., \citealt{gourio2012disaster}, \citealt{croce2014long}). Our theorem implies that the impulse responses with respect to productivity shocks are essentially the same as in models with no risk at all. This result suggests that modeling asset prices in production economies requires either addressing business cycle comovement, or removing the assumptions for our theorem.

We show that a similar irrelevance holds for a broader family of preferences such as smooth ambiguity (\citealt*{klibanoff2005smooth}) or multiplier preferences (\citealt*{hansen2001robust}). Like EZ preferences, these preferences also aim to separate intertemporal substitution from other dimensions of preferences.

To address concerns about the approximate linearity assumption, we examine irrelevance in models solved with projection methods (\citealt{caldara2012computing}). We find that irrelevance holds in these solutions, even for a risk aversion of 100 or a volatility of volatility that is 80\% of the baseline volatility.

\paragraph{Breaking and Adapting to Irrelevance}

Our theorem characterizes when, and in what way, risk modeling is irrelevant. It thus provides a framework for understanding how to generate a meaningful role for risk in model dynamics. This can be done by either ``breaking'' irrelevance (as in \citealt{DiTellaHall2022}) or ``adapting'' to it (as in \citealt{basu2017uncertainty}).

By breaking irrelevance, we mean removing one of the three key assumptions of the theorem. One can remove (1) the separation of intertemporal and risk attitudes (as in \citealt{campbell1999force}), (2) the separation of first and higher moments in the remaining structure (as in \citealt{van2006learning}; \citealt{decker2016market}; \citealt{colacito2014bkk}; \citealt{DiTellaHall2022}), or (3) the approximate linearity of the constraints (as in \citealt{brunnermeier2014macroeconomic}; \citealt{GourioNgo2020ZLB}). All three involve non-trivial modifications of standard models. But all three can have compelling economic motivations. Indeed, incorporating risk-return tradeoffs in production is quite natural.

Alternatively, one can adapt to irrelevance: stay within the class of models covered by the theorem, but drive dynamics with higher moment states.  Examples of adapting to irrelevance include \citet{basu2017uncertainty}, who model shocks to volatility and maintain comovement with nominal rigidities, and \citet{Basu2024Risky}, who model shocks to risk aversion but maintain comovement with a novel investment reallocation channel. 

\paragraph{Related Literature}

\citet*{backus2015risk} is the closest to our paper in providing irrelevance results for several models of time-varying risk and time-varying ambiguity. However, as is common in the literature, they misinterpret the Tallarini property as a separation between the dynamics of quantities and prices. We extend their results into an abstract model with an arbitrary structure and an arbitrary set of shocks. This generality allows us to clarify that the separation is between means and volatilities.

Our theorem builds on risk-adjusted affine approximations (\citealt{jermann1998asset}, \citealt{malkhozov2014asset}, and \citealt*{lopez2018risk}). These approximations capture higher order terms produced by perturbation methods and provide accurate solutions to models with habits or time-varying disaster risk. Indeed, special cases of our irrelevance result are seen in the second and third order perturbation solutions of \citet*{schmitt2004solving} and \citet{van2012term}. This accuracy stems in part from modeling non-Gaussian shocks using entropy, which has been shown to effectively summarize asset pricing properties of pricing kernels (\citealt{backus2014sources}).

Section \ref{sec:ezmodel} uses a two-period model to illustrate irrelevance and discuss implications. Section \ref{sec:tallarini_extension} examines a dynamic model to address asset price volatility and \citet{tallarini2000risk} separation. The main theorem is found in Section \ref{sec:Model}. Section \ref{sec:Numerical-Verification} examines nonlinearities with projection methods. Section \ref{sec:Conclusion} concludes.

\section{Irrelevance Intuition and Implications\label{sec:ezmodel}}

We illustrate the economics behind the irrelevance result with a two-period model.

A household has \citet{epstein1989substitution} utility over consumption for two periods, given by
\begin{align}
    V(c, c_1) =  \exp\left(\left(1-\frac{1}{\psi}\right)c\right)+\exp\left(\left(1-\frac{1}{\psi}\right)
    w_\gamma\left(c_{1}\right)
    \right)
    \label{eq:twoperiod_utility}
\end{align}
where $c$ and $c_1$ are log consumption this period and next period, $\psi$ is the elasticity of intertemporal substitution (EIS) parameter, $w_\gamma(c_{1})$ is the certainty equivalent of $c_1$ defined as
\begin{align}
    w_\gamma\left(c_{1}\right)
    &=\frac{1}{1-\gamma}\log \mathbb{E}\left[\exp\left((1-\gamma)c_{1}\right)\right]
    ,
    \label{eq:twoperiod_g_def}
\end{align}
and $\gamma$ is the risk aversion parameter.\footnote{Throughout the paper, we use lowercase to indicate logs\textemdash that is, $x_{t}\equiv\log X_{t}$.} 

Production follows a linear approximation of a simplified neoclassical model with a capital-only production function and no depreciation:
\begin{align}\label{eq:twoperiod_technology_1}
    c   &= x + \alpha k - i \\
    k_1  &= k + i \label{eq:twoperiod_technology_2}
    \\
    c_1  &= x + \varepsilon_1 + \alpha k_{1} \label{eq:twoperiod_technology_3}
    \\
    r_1  &= x + \varepsilon_1 - (1-\alpha)k_1 \label{eq:twoperiod_technology_4}
    \\    
    \varepsilon_1 &\sim N(0, \sigma^2)
    \label{eq:twoperiod_technology_5} 
\end{align}
where $\alpha$ is production curvature, $k$ is log capital, $i$ is log investment, $r_1$ is the log return on capital, $x$ is log productivity. $\varepsilon_1$ is a productivity shock and $\sigma$ is its volatility.  To understand the linearization, approximate the standard resource constraint following \citet{uhlig1999toolkit}:
\begin{align}
    C &= X K^{\alpha} - I\\
    \Rightarrow \bar{C}(1+c)
      &\approx \bar{X}\bar{K}^{\alpha}\left(1+x+\alpha k\right) 
        - \bar{I}\left(1+i\right)
      \\
    \Rightarrow \bar{C}c 
      &\approx \bar{X}\bar{K}^{\alpha}\left(x+\alpha k\right) 
        - \bar{I}i
        + \text{Constant Terms},
        \label{eq:twoperiod_uhlig}
\end{align}
and then, for ease of exposition, setting approximation points to 1 (with appropriate rescaling of variables) and dropping Constant Terms yields Equation \eqref{eq:twoperiod_technology_1}.

A planner maximizes Equation \eqref{eq:twoperiod_utility} subject to production and resource constraints. The first-order conditions lead to the Euler equation:
\begin{align}
    0	&=\log \mathbb{E}[\exp\left(m_{1}+r_{1}\right)]\label{eq:twoperiod_euler}, \\
    \text{where} \quad m_{1}	&= 
      -\frac{1}{\psi}\Delta c_{1}
      +\left(\frac{1}{\psi}-\gamma\right)\left[
        c_{1}-w_\gamma\left(c_{1}\right)
    \right]
    \label{eq:twoperiod_sdf}
\end{align}
is the log intertemporal marginal rate of substitution of consumption, also known as the log stochastic discount factor (SDF) for the economy, and $\Delta c_1 = c_1 - c$. This SDF is the  \citet{epstein1989substitution} SDF, with the value function next period replaced by consumption next period. The Euler Equation \eqref{eq:twoperiod_euler}, combined with the utility specification (Equations \eqref{eq:twoperiod_utility} and \eqref{eq:twoperiod_g_def}) and the technological constraints (Equations \eqref{eq:twoperiod_technology_1} - \eqref{eq:twoperiod_technology_5}) are the equilibrium conditions.

\subsection{Irrelevance in the Two-Period Model}\label{sec:twoperiod:irrelevance}

EZ preferences separate attitudes towards substitution across time (Equation \eqref{eq:twoperiod_utility}) from attitudes towards substitution across states in the future (Equation \eqref{eq:twoperiod_g_def}), or risk aversion. This separation, in turn, leads to separation of the Euler equation. To see this, use the following property of the certainty equivalent:
\begin{align}
    w_\gamma \left(c_{1}\right) = \mathbb{E}\left[c_1\right] + w_\gamma\left(\Delta \mathbb{E}_{1}c_{1}\right)
\end{align}
where $\Delta \mathbb{E}_{1}$ is the innovation operator ($\Delta \mathbb{E}_{1}c_1 \equiv c_1 - \mathbb{E}[c_1]$). Applying this property to Equation \eqref{eq:twoperiod_euler} yields:
\begin{align}
    0	&= \underbrace{\mathbb{E}\left[-\psi^{-1}\Delta c_1+r_1\right]}_{\text{Intertemporal Term}}
     + \underbrace{
        \left(\frac{1}{\psi}-\gamma\right)w_\gamma\left(\Delta \mathbb{E}_{1}c_{1}\right)
    -\log \mathbb{E}\left[\exp\left(-\gamma\Delta \mathbb{E}_{1}c_{1}+\Delta \mathbb{E}_{1}r_{1}\right)\right]
     }_{\text{Risk Term}}
     .
    \label{eq:twoperiod_euler_separated}
\end{align}
Equation \eqref{eq:twoperiod_euler_separated} separates the Euler equation into an Intertemporal Term and a Risk Term. 

The Intertemporal Term characterizes expected tradeoffs. The household has preferences over the smoothness of the consumption profile over time, in expectation, that determine the tradeoff of consumption today for expected consumption tomorrow given the expected return on capital. The strength of this preference is driven by the EIS, captured by $\psi$. 

Risk aversion, by definition, does not enter in this intertemporal tradeoff. Instead, it enters through the Risk Term, which captures tradeoffs across \emph{un}expected states in the future. With EZ preferences, risk aversion is fully captured by $\gamma$.

In principle, the Intertemporal and Risk Terms interact through the production technology. In the real world, producing goods involves risk-return tradeoffs. But in standard models, transforming resources across time (Equations \eqref{eq:twoperiod_technology_1}-\eqref{eq:twoperiod_technology_4}) is separated from transformation across unexpected states in the future (Equation \eqref{eq:twoperiod_technology_5}). In a sense, the stochastic growth model is a perfect foresight model with shocks ``tacked on.''\footnote{%
We thank an anonymous referee for suggesting this analogy.
} %
As a result, the model structure in this linearized framework does not lead to a direct interaction between the Intertemporal and Risk Terms.

Approximate linearity implies there are no interactions at all, not even indirect ones. To see this, take innovations of Equations \eqref{eq:twoperiod_technology_3} and \eqref{eq:twoperiod_technology_4}:
\begin{align}
    \Delta \mathbb{E}_{1}c_{1} &= \Delta \mathbb{E}_{1} \left[x+\varepsilon_1 + \alpha k_{1}\right] = \varepsilon_1, \label{eq:innovations_algebra}
    \quad \text{and} \quad
    \Delta \mathbb{E}_{1}r_{1} = \Delta \mathbb{E}_{1} \left[x+\varepsilon_1 - (1-\alpha)k_{1}\right] = \varepsilon_1,    
\end{align}
respectively, where $x$ and $k_1$ are removed since they are known. Plugging into Equation \eqref{eq:twoperiod_euler_separated} and using the properties of the log-normal distribution yields
\begin{align}
    0	&= \underbrace{
        \vphantom{\frac{1}{2}}\mathbb{E}\left[-\frac{1}{\psi}\Delta c_1+r_1\right]
    }_{\text{Intertemporal Term}} 
     + \underbrace{
        \frac{1}{2}\left(\frac{1}{\psi}-1\right)\left(1-\gamma\right)\sigma^2
     }_{\text{Risk Term}}
     . 
     \label{eq:twoperiod_euler_separated_b}
\end{align}
Equation \eqref{eq:twoperiod_euler_separated_b} shows that variables that govern intertemporal transformation (namely $x$ and $k$) are irrelevant for the Risk Term.

As a result of this stark separation, the Risk Term drops out if we differentiate Equation \eqref{eq:twoperiod_euler_separated} with respect to $x$:
\begin{align}
    0
    &=
    \frac{d}{d x} \mathbb{E}\left[-\frac{1}{\psi}\Delta c_1+r_1\right].
\end{align}
Plugging in $\mathbb{E}[\Delta c_1] = (\alpha + 1) i$ (based on Equations \eqref{eq:twoperiod_technology_2}-\eqref{eq:twoperiod_technology_3}) and $\mathbb{E} [r_1 ]= x - (1-\alpha)(k+i)$ (based on Equations \eqref{eq:twoperiod_technology_2} and \eqref{eq:twoperiod_technology_4}), we can solve for $\frac{di}{dx}$ in closed form:
\begin{align}
    \frac{di}{dx}
      &=
      \frac{\psi}{(\alpha+1) + \psi(1-\alpha)}.
    \label{eq:di/dx}
\end{align}
Thus, $\gamma$ and $\sigma$ are irrelevant for the elasticity of investment with respect to productivity. Similarly, one can show that $\frac{di}{dk}$, $\frac{dc}{dx}$, and $\frac{dc}{dk}$ also do not depend on $\gamma$ or $\sigma$. 

What these elasticities have in common is that they are all with respect to ``first moment states''---variables that affect the expected values of the model structure next period but not higher moments. Thus, $\gamma$ and $\sigma$ are irrelevant for the elasticity of endogenous variables with respect to first moment states. This irrelevance extends to a much broader class of models, as we show in Sections \ref{sec:tallarini_extension}-\ref{sec:Numerical-Verification}.

As a result of this irrelevance, the presence of risk has no effect on how the model responds to first moment shocks. The special cases of the model with perfect foresight (or absence of risk, $\sigma=0)$ or risk neutrality ($\gamma=0$) lead to identical impulse responses to shocks to first moment states. 

An important caveat is that $\gamma$ and $\sigma$ still affect the overall endogenous variable levels, such as the level of $i$, as well as the response of these variables to high moment shocks, such as $\frac{di}{d\sigma}$. This means that risk is relevant for steady states, and indeed its effects on steady states can be significant, and the dynamics of endogenous variables through high moment states. 

\subsection{Breaking Irrelevance}\label{sec:twoperiod:breaking}

To have risk play a meaningful role in business cycles, one can ``break'' irrelevance. By breaking irrelevance, we mean having $\gamma$ or $\sigma$ enter into $\frac{di}{dx}$ (Equation \eqref{eq:di/dx}). This can be achieved by removing any of the three key assumptions:
\begin{enumerate}
    \item \emph{Separation of intertemporal and risk preferences}: Epstein-Zin preferences (Equations \eqref{eq:twoperiod_utility} and \eqref{eq:twoperiod_g_def}) ensure $\gamma$ does not enter the Intertemporal Term of the Euler equation (Equation \eqref{eq:twoperiod_euler_separated}). Replacing these preferences with many alternatives would then have $\gamma$ enter into $di/dx$. 

    However, having $\gamma$ enter the Intertemporal Term implies that $\gamma$ should no longer be interpreted solely as risk aversion. For example, if we replace EZ preferences with power utility, then $\gamma$ should be interpreted as \emph{both} risk aversion and the EIS, and it is the EIS interpretation of $\gamma$ that is relevant for Equation \eqref{eq:di/dx}.

    There are more economically meaningful ways to remove the separation between intertemporal and risk preferences. In the \citet{campbell1999force} habit model, risk aversion (and the EIS) vary over time as consumption gets closer to habit. Having heterogeneous households leads to similar effects, as shown by \citet{guvenen2009parsimonious}.  Both models lead to links between the Intertemporal and Risk Terms of the Euler equation, breaking irrelevance.

    \item \emph{Separation of first and higher moments of the remaining model structure}: The production technology (Equations \eqref{eq:twoperiod_technology_1} - \eqref{eq:twoperiod_technology_4}) assumes purely intertemporal transformation (using $i$ to move goods across time) is separated from transformation across unexpected states in the future (controlled by $\sigma$). Removing this assumption leads to interactions between the Intertemporal and Risk Terms, implying $\gamma$ and $\sigma$ enter into $\frac{di}{dx}$.\footnote{%
    For example, suppose production is non-separable as follows: $k_1=(k+i)\mathbf{1}\{\varepsilon_1 \ge 0\}$. Then consumption innovations \eqref{eq:innovations_algebra} gain investment terms
    \[
    \Delta \mathbb{E}_1 c_1
    = \varepsilon_1
    + \alpha\left((k+i)\mathbf{1}\{\varepsilon_1 \ge 0\}-\mathbb{E}[(k+i)\mathbf{1}\{\varepsilon_1 \ge 0\}]\right),
    \]
    and thus investment shows up in the Risk Term of Equation \eqref{eq:twoperiod_euler_separated}. Differentiating with respect to $x$ then leads to risk terms $\gamma$ and $\sigma$ entering $\frac{di}{dx}$, breaking irrelevance.}

    In the general model (Section \ref{sec:Model}), this assumption regards separation in all sectors of the economy, not just final goods production. This includes the government sector, the financial sector, and the labor market---essentially, any aspect of the model other than preferences. It also applies to the endowment process in endowment economies. For example, \citet{bansal2004risks} assume separation of first and higher moments in long-run growth, leading to the irrelevance of risk aversion for how asset prices respond to expected growth shocks (see their Equation (5)). 

    Removing this assumption is natural. In the real world, production decisions have risk-return tradeoffs. In the literature, such tradeoffs are studied in \citet{belo2010production}, \citet{cochrane2021rethinking}, \citet{DiTellaHall2022}, and \citet{Zeke2024RiskTaking}, for instance. In all of these models, production must trade off transforming goods across time versus transforming goods across unexpected states. For example, in \citeauthor{DiTellaHall2022} workers and entrepreneurs produce goods, with entrepreneurs facing uninsurable idiosyncratic shocks, leading to an aggregate production function that incorporates risk. A common theme in these models is that higher moments such as volatilities are endogenous. Naturally, then, both risk aversion and volatility affect how the model responds to first moment shocks.

    \item \emph{Linearity of the constraints}:  If the constraints are nonlinear, then Equation \eqref{eq:twoperiod_uhlig} would require higher order terms like $x^2$, implying $c_1$ has interactions between $\varepsilon_1$ and $x$. This would lead to $x$ appearing in the innovation $\Delta \mathbb{E}_{1}c_1$ (Equation \eqref{eq:innovations_algebra}), implying the Risk Term of Equation \eqref{eq:twoperiod_euler_separated_b} depends on $x$, breaking irrelevance.

    The economics are intuitive. Non-linear relationships link the current state and forward looking volatility (\citealt{chen2017external}). A priori, it is hard to tell whether the constraints are nonlinear enough to meaningfully break irrelevance. In Section \ref{sec:Numerical-Verification} we show that simply increasing risk aversion or the volatility of the volatility does not lead to enough nonlinearity, but it is difficult to provide a complete characterization. 
    
    Nevertheless, the literature shows that occasionally binding constraints can introduce sufficient nonlinearity. For example, \citet{GourioNgo2020ZLB} show irrelevance breaks when the nominal short-term interest rate is close to its zero lower bound as a result of increased endogenous volatility. \citeauthor{brunnermeier2014macroeconomic}'s  \citeyearpar{brunnermeier2014macroeconomic} model of financial frictions shows how nonlinearities dramatically increase the volatility of variables far from the steady state, much more so than found in the more linear predecessor \citet{Bernanke1999}. 

\end{enumerate}

\subsection{Adapting to Irrelevance}\label{sec:twoperiod:adapting}

Alternatively, one can incorporate risk in model dynamics by ``adapting'' to irrelevance---that is, staying within the implications of irrelevance (e.g. Equation \eqref{eq:di/dx}), but driving the model 
with shocks to ``higher moment states''---variables that govern the higher central moments of the model structure.

Adapting to irrelevance requires to avoid the \citet{barro1984time} comovement problem. \citet{barro1984time} show that, for neoclassical models, many shocks fail to generate the observed comovement of consumption, investment, and labor. This limitation can be avoided by incorporating additional model elements.

Examples of adapting to irrelevance include \citet{basu2017uncertainty} and \citet{Basu2024Risky}. \citet{basu2017uncertainty} drive business cycles with shocks to the volatility of time preference, and maintain comovement by incorporating nominal rigidities. \citet{Basu2024Risky} drive business cycles with shocks to risk aversion, and recover comovement by modeling the reallocation of savings between labor and capital. Other methods for overcoming the comovement problem are studied in \citet{dupor2014analytics}, including consumption-investment complementarities and externalities in leisure preferences.

\section{Irrelevance in a Dynamic Model}\label{sec:tallarini_extension}

We examine a stochastic growth model with shocks to the volatility of productivity and capital adjustment costs. This model allows for dynamic asset prices, clarifying and extending the \citet{tallarini2000risk} separation result, and explaining a key difficulty of modeling asset prices in production economies.

A representative household has \citet*{epstein1989substitution} recursive
preferences
\begin{align}
V_{t} & =\left\{ (1-\beta)C_{t}^{1-1/\psi}+\beta W_{t}^{1-1/\psi}\right\} ^{1/(1-1/\psi)},\label{eq:ezmodel_v_recur}\\
\text{where} \quad W_{t} & =\left[\mathbb{E}_{t}(V_{t+1}^{1-\gamma})\right]^{\frac{1}{1-\gamma}},\label{eq:ezmodel_certequiv}
\end{align}
$W_{t}$ is the certainty equivalent, $V_{t}$ is the utility aggregator, $C_{t}$ is consumption, 
$\psi$ is the elasticity of intertemporal substitution of consumption
(EIS), $\gamma$ is risk aversion, and $\beta$ is the time preference parameter. Production, which in equilibrium is used for consumption and investment $I_t$, follows
\begin{align}
C_{t}+I_{t} & =X_{t}K_{t}^{\alpha}N_{t}^{(1-\alpha)},\label{eq:ezmodel_rc}\\
K_{t+1} & =(1-\delta)K_{t}+\Phi(I_{t}/K_{t})K_{t},
\end{align}
where $X_{t}$ is productivity, $K_{t}$ is
capital, $N_{t}$ is labor, $\alpha$ is the capital share, $\delta$
is depreciation, and $\Phi(I_{t}/K_{t})$ is, the adjustment cost in capital accumulation, follows \citet{jermann1998asset}: 
\begin{align*}
\Phi(I_{t}/K_{t})\equiv & a_{1}+\frac{a_{2}}{1-1/\xi}(I_{t}/K_{t})^{1-1/\xi},
\end{align*}
where $\xi$ is the elasticity of investment with respect to Tobin's $q$, and the constants $a_{1}=\frac{\overline{IK}}{1-1/\xi}$ and $a_{2}=\overline{IK}^{1/\xi}$ are chosen so that, at the non-stochastic steady state investment rate $\overline{IK}$, the adjustment cost is zero and the marginal cost of investment is 1. With no disutility from labor, in equilibrium $N_{t}$ will always be the maximum possible, which we normalize to 1.

Log-productivity $x_{t}$ follows a heteroskedastic AR(1) process
in logs: 
\begin{align}
x_{t+1} & =\phi_{x}x_{t}+\epsilon_{x,t+1},\label{eq:ezmodel_x_lom}\\
\epsilon_{x,t+1} & \sim N\left(0,\frac{1}{2}\sigma_{x}^{2}h_{t}\right)\quad\text{i.i.d.,}
\end{align}
where $\phi_{x}$ is the persistence of productivity, and $\epsilon_{x,t+1}$ is a shock that depends on the volatility level $\sigma_{x}$ and the volatility process $h_{t}$. $h_{t}$ follows the AR(1) process
\begin{align*}
h_{t+1} & =\phi_{h}h_{t}+\epsilon_{h,t+1},\\
\epsilon_{h,t+1} & \sim N(0,\sigma_{h}^2)\quad\text{i.i.d.,}
\end{align*}
where $\phi_{h}$ is the persistence, and $\epsilon_{h,t+1}$ is a
shock with volatility $\sigma_{h}$.

Optimal firm investment implies the Euler equation
\begin{align}
1= & \mathbb{E}_{t}\left[M_{t+1}R_{t+1}\right],\label{eq:ezmodel_euler_level}
\end{align}
where $M_{t+1}$ is the Epstein-Zin stochastic discount factor
\begin{align*}
M_{t+1}\equiv & \beta\left(\frac{C_{t+1}}{C_{t}}\right)^{-1/\psi}\left(\frac{V_{t+1}}{W_{t}}\right)^{-(\gamma-1/\psi)},\label{eq:ezmodel_SDF}
\end{align*}
and $R_{t+1}$ is the investment return
\begin{align*}
R_{t+1}\equiv & \Phi'(I_{t}/K_{t})\left\{ \alpha X_{t+1}K_{t+1}^{-(1-\alpha)}+\frac{1}{\Phi'(I_{t+1}/K_{t+1})}\left[\Phi(I_{t+1}/K_{t+1})+1-\delta\right]-\frac{I_{t+1}}{K_{t+1}}\right\} .
\end{align*}

As the production technology is homogeneous of degree one, the stock price $S_{t}$
can be obtained using Q-theory (\citealt{hayashi1982tobin}, to be
\citealt{abel1983optimal}):
\begin{align}
    S_{t}= & Q_{t}K_{t}, \quad \text{where} \quad Q_{t}=  \Phi'(I_{t}/K_{t})^{-1}
\end{align}        
is Tobin's Q.

In this setting, the state variables are $K_{t}$, $x_{t}$, and $h_{t}$. Only $h_{t}$ is a higher moment state. $K_{t}$ and $x_{t}$ are first moment states.

\subsection{Risk-Adjusted Affine Solution}

Our risk-adjusted affine solution proceeds in two steps. First, we transform the Equations (\ref{eq:ezmodel_certequiv})
and (\ref{eq:ezmodel_euler_level}) into logs, preserving exact equality:
\begin{align}
w_{t} & =\frac{1}{1-\gamma}\log\mathbb{E}_{t}\exp\left[(1-\gamma)v_{t+1}\right],\label{eq:ezmodel_logCE}\\
0 & =\log\mathbb{E}_{t}[\exp\left(m_{t+1}+r_{t+1}\right)].\label{eq:ezmodel_euler}
\end{align}
Second, we log-linearize the remaining constraints following \citet{uhlig1999toolkit}:
\begin{align}
    \bar{V}^{1-\frac{1}{\psi}}v_{t}= & (1-\beta)\bar{C}^{1-\frac{1}{\psi}}c_{t}+\beta\bar{W}^{1-\frac{1}{\psi}}w_{t}\label{eq:app_ezlin_v_recur_main}\\    
    \overline{CY}c_{t}+\overline{IY}i_{t}= & x_{t}+\alpha k_{t}\label{eq:app_ezlin_rc_main}\\
    k_{t+1}= & (1-\delta)k_{t}+\delta i_{t}\\
    \bar{R}r_{i,t+1}= & (\alpha\overline{YK}+1-\bar{R})+\alpha\overline{YK}(x_{t+1}-(1-\alpha)k_{t+1})+q_{t+1}-\bar{R}q_{t}\label{eq:app_ezlin_rinv_main}\\
    q_{t}= & \xi^{-1}(i_{t}-k_{t}-\delta)\label{eq:app_ezlin_q_main}\\    
    s_{t}= & q_{t}+k_{t}
\end{align}
Details of the linearization are found in Appendix \ref{subsec:Appendix_ezmodel}.

Although Equations \eqref{eq:ezmodel_logCE} and \eqref{eq:ezmodel_euler} are non-linear, the normal shocks imply
\begin{align}
    \log\mathbb{E}_{t}\left[\exp(\epsilon_{x,t+1})\right]=\frac{1}{2}\sigma_{x}^{2}h_{t},
\end{align}
and thus the solution is linear in state variables.

\subsection{Irrelevance\label{subsec:ez_Irrelevance}}

The method of undetermined coefficients implies the law of motion for the
log of stock prices has the form
\begin{align}
s_{t} & =Z_{s0}\left(\beta,\psi,\gamma,\alpha,\delta,\xi,\sigma_{x},\sigma_{h},\boldsymbol{\bar{A}}\right)+Z_{sk}\left(\psi,\alpha,\delta,\xi,\boldsymbol{\bar{A}}\right)k_{t}\nonumber \\
 & \hphantom{\{=\}}+Z_{sx}\left(\psi,\alpha,\delta,\xi,\phi_{x},\boldsymbol{\bar{A}}\right)x_{t}+Z_{sh}\left(\beta,\psi,\gamma,\alpha,\delta,\xi,\phi_{h},\sigma_{x},\boldsymbol{\bar{A}}\right)h_{t},\label{eq:eq:ezmodel_s_lom}
\end{align}
and capital has the form
\begin{align}
k_{t+1} & =Z_{k0}\left(\beta,\psi,\gamma,\alpha,\delta,\xi,\sigma_{x},\sigma_{h},\boldsymbol{\bar{A}}\right)+Z_{kk}\left(\psi,\alpha,\delta,\xi,\boldsymbol{\bar{A}}\right)k_{t}\nonumber \\
 & \hphantom{\{=\}}+Z_{kx}\left(\psi,\alpha,\delta,\xi,\phi_{x},\boldsymbol{\bar{A}}\right)x_{t}+Z_{kh}\left(\beta,\psi,\gamma,\alpha,\delta,\xi,\phi_{h},\sigma_{x},\boldsymbol{\bar{A}}\right)h_{t},\label{eq:ezmodel_k_lom}
\end{align}
where $\boldsymbol{\bar{A}}$ is a set of approximation points. The
full expressions for the coefficients, along with an outline of the method of undetermined coefficients, is in Appendix \ref{subsec:app_ezmodel_Solving-the-model}.

Irrelevance is seen in the elasticities with respect to capital and productivity: $Z_{sk}\left(\psi,\alpha,\delta,\xi,\boldsymbol{\bar{A}}\right)$, $Z_{sx}\left(\psi,\alpha,\delta,\xi,\phi_{x},\boldsymbol{\bar{A}}\right)$, $Z_{kk}\left(\psi,\alpha,\delta,\xi,\boldsymbol{\bar{A}}\right)$, and $Z_{kx}\left(\psi,\alpha,\delta,\xi,\phi_{x},\boldsymbol{\bar{A}}\right)$ are all independent of $\gamma$, $\phi_h$, and $\sigma_h$. Thus, risk aversion and time-varying risk are irrelevant for how the model responds to changes in productivity and capital.

In the special case where $h_{t}$ is a constant, we obtain
the \citet{tallarini2000risk} result: risk aversion is irrelevant for the model's quantity dynamics. However, risk aversion does matter for the intercepts in Equations (\ref{eq:eq:ezmodel_s_lom}) and (\ref{eq:ezmodel_k_lom}),
consistent with \citet{tallarini2000risk}'s finding that the equity premium depends on risk aversion.

Equation (\ref{eq:eq:ezmodel_s_lom}) goes beyond \citet{tallarini2000risk},
however. When $h_{t}$ is constant, risk aversion
is not only irrelevant for quantity dynamics, but also irrelevant
for the model's stock price dynamics. More generally, Section \ref{sec:Model}
shows that constant $h_{t}$ implies risk aversion is irrelevant for
the dynamics of all asset prices, not just quantities. This clarifies that the \citet{tallarini2000risk} separation result is not between macroeconomic and financial variables, but between first and higher moments of the model. Risk aversion affects the means of variables, but not their volatilities or impulse responses to shocks.

Even if $h_{t}$ varies over time, Equation (\ref{eq:eq:ezmodel_s_lom}) shows that the elasticity of the stock price with respect to productivity is independent of risk aversion and time-varying risk. This result implies that the inclusion of elements such as long-run risk or disaster risk has little effect on how stock prices respond to productivity shocks. 

\section{The General Irrelevance Theorem\label{sec:Model}}

We prove the irrelevance result in a general, abstract structure that covers a broad class of representative-agent rational-expectations economic models used in the literature. Section \ref{subsec:Structure} describes this structure, Section \ref{subsec:Approximation} presents the risk-adjusted approximation that we use to prove the theorem,  Section \ref{subsec:Main-Theorem} characterizes our irrelevance theorem, and Section \ref{subsec:TheoremExtension} shows a model extension in which irrelevance still holds.

\subsection{Model Structure \label{subsec:Structure}}

The structure is characterized by a system of expectational equations (equilibrium conditions) of the form
\begin{eqnarray}
\mathbb{E}\left[f\left(\mathbf{z}_{t+1},\mathbf{x}_{t+1}|\mathbf{z}_{t},\mathbf{x}_{t},\mathbf{h}_{t},\mathbf{z}_{t-1};\Theta\right)\right]=0,\label{eqn:ExpEqnGeneral1}
\end{eqnarray}
where $f(\cdot)$ is a vector of $n_{z}$ arbitrary scalar functions, $\text{\textbf{z}}_{t}$ is a vector of $n_{z}$ endogenous variables, $\text{\textbf{x}}_{t}$ is a vector of $n_{x}$ ``first moment'' exogenous variables, $\text{\textbf{h}}_{t}$ is a vector of $n_{h}$ ``higher moment'' exogenous variables, and $\Theta$ is a set of model parameters. The sense in which $\text{\textbf{x}}_{t}$ are first moment variables and $\text{\textbf{h}}_{t}$ are higher moment variables will become clear below. Solving the model consists in finding the process of motion for $\text{\textbf{z}}_{t}$ that satisfies Equation (\ref{eqn:ExpEqnGeneral1}) given exogenous process for $\text{\textbf{x}}_{t}$ and $\text{\textbf{h}}_{t}$.

The model parameters are
\begin{align*}
\Theta\equiv & \{\gamma,\theta_{h},\theta_{0}\},
\end{align*}
where $\gamma$ is the coefficient of risk aversion, $\theta_{h}$ is a vector of parameters that govern higher moments, and $\theta_{0}$ is a vector that contains all other model parameters. Our goal is to show that $\gamma$ and $\theta_{h}$ are irrelevant for key properties of the model. In particular, we want to show that the endogenous reaction of $\text{\textbf{z}}_{t}$ to the first moment variables $\text{\textbf{x}}_{t}$ and to lagged variables $\text{\textbf{z}}_{t-1}$ is not affected by risk aversion and by the higher order moment parameters.

\subsubsection{Specialized Notation\label{subsec:Notation}}

Here we introduce some specialized notation to keep the formalism compact. We define $\text{\textbf{exp}}(\textbf{u})$ as a vectorized version of the exponential function. This function maps an $n_{u}$-vector $\textbf{u}=(u_{1},u_{2},...,u_{n_{u}})^{\top}$ of random variables into a vector where each element is the exponential of its corresponding element in $\textbf{u}$; that is,
\begin{equation}
\text{\textbf{exp}}(\textbf{u})\equiv\left(\exp(u_{1}),\exp(u_{2}),...,\exp(u_{n_{u}})\right)^{\top}.\label{eq:main_exp_vec}
\end{equation}
Similarly, we define $\boldsymbol{\log}(\textbf{u})$ as
\begin{equation}
\boldsymbol{\log}(\textbf{u})\equiv\left(\log(u_{1}),\log(u_{2}),...,\log(u_{n_{u}})\right)^{\top}.\label{eq:main_log_vec}
\end{equation}
Finally, we define entropy as the conditional matrix operator %
\footnote{
This operator is known as relative entropy in the asset pricing literature
or Theil's second entropy measure in the literature on inequality.
}%
\begin{eqnarray}
\mathbb{\mathbb{H}}_{t}[\textbf{u}]\equiv\boldsymbol{\log}\thinspace\thinspace\mathbb{E}_{t}\left[\text{\textbf{exp}}\left(\textbf{u}-\mathbb{E}_{t}[\textbf{u}]\right)\right],\label{eqn:curvature}
\end{eqnarray}
which allows us to express the vectorized cumulant generating function (CGF) $K_u(A)$ for any random vector $\textbf{u}$ and matrix $A$ of appropriate dimensions as
\begin{eqnarray}
	K_u(A)
	\equiv \boldsymbol{\log}\,\mathbb{E}_t\!\left[\textbf{exp}(A\mathbf{u})\right]= \boldsymbol{\log}\,\mathbb{E}_t\!\left[\textbf{exp}\!\left(A\mathbb{E}_t[\mathbf{u}]
		  + A(\mathbf{u}-\mathbb{E}_t[\mathbf{u}])\right)\right]
	= A\mathbb{E}_t[\mathbf{u}] + \mathbb{H}_t[A\mathbf{u}]. \label{eq:app_main_Hdecomp}
\end{eqnarray}

\subsubsection{Exogenous Processes\label{subsec:Exogenous-Processes}}

First moment states $\text{\textbf{x}}_{t}$ follow the heteroskedastic autoregressive process
\begin{align}
\mathbf{x}_{t+1} & =\mu_{x}+\Phi_{x}\mathbf{x}_{t}+\epsilon_{x,t+1},\label{eq:main_x_process}\\
\epsilon_{x,t+1} & \sim\text{pdf}_{x}\text{\ensuremath{(\cdot|\text{\textbf{h}}_{t},\theta_{0})}},\label{eq:main_epx_pdf}
\end{align}
where $\mu_{x}\in\theta_{0}$ is a vector of intercepts, $\Phi_{x}\in\theta_{0}$ is a matrix of autoregressive coefficients, and $\text{pdf}_{x}\text{\ensuremath{(\cdot|\text{\textbf{h}}_{t},\theta_{0})}}$ is a probability distribution function (pdf) with zero mean. 

The higher moment states $\text{\textbf{h}}_{t}$ follow the autoregressive process
\begin{align}
\text{\textbf{h}}_{t+1} & =\mu_{h}+\Phi_{h}\mathbf{h}_{t}+\epsilon_{h,t+1},\label{eq:main_h_process}\\
\epsilon_{h,t+1} & \sim\text{pdf}_{h}\text{\ensuremath{(\cdot|\text{\ensuremath{\text{\textbf{h}}_{t},}}\theta_{h}),}}\label{eq:main_eph_pdf}\\
\epsilon_{h,t+1},\text{ and }\epsilon_{x,t+1} & \quad\text{are independent},\label{eq:independent_hx}
\end{align}
where $\mu_{h}\in\theta_{h}$ is a vector of intercepts, $\Phi_{h}\in\theta_{h}$ is a matrix of autoregressive coefficients, and $\text{pdf}_{h}\text{\ensuremath{(\cdot|\text{\textbf{h}}_{t},\theta_{h})}}$ is an arbitrary pdf. The independence of $\epsilon_{h,t+1}$ and $\epsilon_{x,t+1}$ implies first and higher moment states do not interact by assumption.

Equations (\ref{eq:main_x_process})-(\ref{eq:independent_hx}) imply that $\theta_{h}$ governs higher moments; that is, $\theta_{h}$ determines the evolution of $\text{\textbf{h}}_{t}$ which, in turn, determines the conditional higher moments of $\text{\textbf{h}}_{t+1}.$ In contrast, $\text{\textbf{x}}_{t}$ does not affect any of the shock distributions, and thus is a vector of first moment states.

\subsubsection{Equilibrium Conditions \label{subsec:EquilibriumConditions}}

We divide the set of equilibrium conditions (\ref{eqn:ExpEqnGeneral1}) into two groups. The first group of equilibrium conditions is given by the subsystem of expectational equations
\begin{eqnarray}
	\mathbb{E}\left[f_{0}\left(\mathbf{z}_{t+1},\mathbf{x}_{t+1}|\mathbf{z}_{t},\mathbf{x}_{t},\mathbf{\mathbf{h}_{t},z}_{t-1};\theta_{0}\right)\right]=0,\label{eq:main_production}
\end{eqnarray}
where $f_{0}(\cdot)$ is a vector of $n_{z}$-$n_{m}$ scalar functions that depend on $\theta_{0}$ only, and then do not involve $\gamma$ or $\theta_{h}$. Equations in this subsystem typically represent production and resource constraints, and may include government policies and other economic agents' optimality conditions. For the model in Section \ref{sec:tallarini_extension}, Equation (\ref{eq:main_production}) includes the production constraints as well as the intratemporal utility aggregator. Equation (\ref{eq:main_production}) also nests endowment economies in which the additional constraints simply set consumption equal to an exogenous endowment.

Importantly, Equation (\ref{eq:main_production}) does not depend on risk aversion $\gamma$ or parameters that govern higher moments $\theta_{h}$.

The second group contains those equations that involve risk aversion, and is represented by
\begin{align}
\boldsymbol{1}_{n_{m}} & =\mathbb{E}_{t}[\text{\textbf{exp}}(m_{t+1}\boldsymbol{1}_{n_{m}}+\mathbf{r}_{t+1})].\label{eq:main_Euler}
\end{align}

Equation (\ref{eq:main_Euler}) is a vector of $n_{m}$ asset-pricing Euler equations that must hold in equilibrium, where $\boldsymbol{1}_{n_{m}}$ is an $n_{m}$-vector of ones, $\mathbf{r}_{t+1}$ is the $n_{m}$-vector of returns, and $m_{t+1}$ is the log stochastic discount factor. Here, $\mathbf{r}_{t+1}$ represents the intertemporal return on any technology, not necessarily an asset return.

The log SDF is given by
\begin{equation}
m_{t+1}=m_{t+1}^{*}+\left(\frac{1}{\psi}-\gamma\right)(v_{t+1}-w_{t}).\label{eq:main_sdf}
\end{equation}
In this equation, $\psi\in\theta_{0}$ is the EIS, $m_{t+1}^{*}$ is the component of the discount factor that is not affected by $\gamma$, and $\left(\frac{1}{\psi}-\gamma\right)(v_{t+1}-w_{t})$ is the recursive preference component of the discount factor, where
\begin{align}
w_{t} & =\frac{1}{1-\gamma}\log\mathbb{E}_{t}\left[\exp((1-\gamma)v_{t+1})\right]\label{eq:main_certainty_eq}
\end{align}
is the log of the certainty equivalent of the representative agent (log) intertemporal utility $v_{t+1}$, as in \citet*{epstein1989substitution}.\footnote{In models with growth, $w_{t}$ and $v_{t+1}$ are usually detrended in order to solve for stationary dynamics. Equation (\ref{eq:main_certainty_eq})
is still valid in this case, assuming that the variables are detrended
and adjusted for growth.} Equations (\ref{eq:main_sdf}) and (\ref{eq:main_certainty_eq}) nest Epstein-Zin preferences with $m_{t+1}^{\ast}=\log\beta-(1/\psi)\Delta c_{t+1}$, and can also accommodate extensions of Epstein-Zin preferences like the time preference shocks of \citet{albuquerque2016valuation}, or the Epstein-Zin-Habit preferences of \citet{yang2015intertemporal}. Indeed, Section \ref{subsec:TheoremExtension} shows that irrelevance holds under an even more general class of preferences.

\subsection{Affine Approximation\label{subsec:Approximation}}

The nonlinear nature of equilibrium conditions, in general, does not allow for analytical proofs. We then use risk-adjusted approximations of these conditions to prove the irrelevance theorem. 

Our approximation is characterized by the following two assumptions:
\begin{assump}
\label{def:cond_affine}(Conditionally Affine Constraints) The model
has conditionally affine constraints if Equation (\ref{eq:main_production})
is approximated by
\begin{align}
b^{1}(\theta_{0})+B_{z}^{1}(\theta_{0})\mathbf{z}_{t}+B_{zl}^{1}(\theta_{0})\mathbf{z}_{t-1}+B_{x}^{1}(\theta_{0})\mathbf{x}_{t} & =\boldsymbol{\log}\thinspace\thinspace\mathbb{E}_{t}\left[\text{\textbf{exp}}\left(D_{z}^{1}(\theta_{0})\mathbf{z}_{t+1}+D_{x}^{1}(\theta_{0})\mathbf{x}_{t+1}\right)\right],\label{eq:main_aff_tech}
\end{align}
where $b^{1}(\theta_{0})$, $B_{z}^{1}(\theta_{0})$, $B_{zl}^{1}(\theta_{0})$, $B_{x}^{1}(\theta_{0})$, \textbf{$D_{z}^{1}(\theta_{0})$}, and $D_{x}^{1}(\theta_{0})$ are conforming vector and matrix functions of $\theta_{0}$ that approximate the vector of $f_{0}$ functions in Equation (\ref{eq:main_production}). \end{assump}

\begin{assump}
\label{def:cond_affine_CGF}(Conditionally Affine CGFs) The cumulant
generating functions of the shocks $\epsilon_{x,t+1}$ in Equation
(\ref{eq:main_epx_pdf}) and $\epsilon_{h,t+1}$ in Equation (\ref{eq:main_eph_pdf})
are approximated by
\begin{align}
K_{\epsilon_x}(\nu_{x})=\boldsymbol{\log}\thinspace\thinspace\mathbb{E}_{t}\left[\boldsymbol{\exp}\left(\nu_{x}\epsilon_{x,t+1}\right)\right] & =\tilde{g}_{x}(\nu_{x};\Theta)+\tilde{G}_{x}(\nu_{x};\Theta)\mathbf{h}_{t},\label{eq:main_aff_epx}\\
\text{and } \, K_{\epsilon_h}(\nu_{h})=\boldsymbol{\log}\thinspace\thinspace\mathbb{E}_{t}\left[\boldsymbol{\exp}\left(\nu_{h}\epsilon_{h,t+1}\right)\right] & =\tilde{g}_{h}(\nu_{h};\Theta)+\tilde{G}_{h}(\nu_{h};\Theta)\mathbf{h}_{t},\label{eq:main_aff_eph}
\end{align}
respectively, where for any $n>0,$ $\nu_{x}$ is any $n\times n_{x}$ matrix, $\nu_{h}$ is any $n\times n_{h}$ matrix, $\tilde{g}_{x}(\nu_{x};\Theta)$ and $\tilde{g}_{h}(\nu_{h};\Theta)$ are $n$-vector-valued functions, and $\tilde{G}_{x}(\nu_{x};\Theta)$ and $\tilde{G}_{h}(\nu_{h};\Theta)$ are $n\times n_{h}$-matrix-valued functions. \end{assump} Equation (\ref{eq:main_aff_tech}) assumes that the constraints in Equation (\ref{eq:main_production}) can be approximated with a risk-adjusted linearization, while Equations (\ref{eq:main_aff_epx}) and (\ref{eq:main_aff_eph}) assume that time-varying higher order properties in the model are well preserved by linearizing the CGFs of the shocks. 

Assumption \ref{def:cond_affine} allows us to clearly describe the dependence structure of the model. Most equilibrium conditions in economic models are of the type represented by Equation (\ref{eq:main_aff_tech}), as they depend only on $\theta_{0}$. Risk aversion affects model dynamics through Equations (\ref{eq:main_Euler})-(\ref{eq:main_certainty_eq}), as they depend on $\gamma$. We do not approximate these equations to effectively capture risk effects as in \citet{malkhozov2014asset}.\footnote{The component $m_{t+1}^{*}$of the discount factor in Equation (\ref{eq:main_sdf})
has to be linearized if necessary to express it as a linear combination
of state variables. This linearization, however, does not affect the
ability of the model approximation to capture risk effects. } These equations contain expectations and then may depend on $\text{\textbf{h}}_{t}$ and $\theta_{h}$ as shown in Equations (\ref{eq:main_aff_epx}) and (\ref{eq:main_aff_eph}).

Without loss of generality, the set of returns $\mathbf{r}_{t+1}$ and the variables $m_{t+1}^{*}$ and $v_{t+1}$ are included in the set of endogenous variables $\mathbf{z}_{t+1}$ so that
\begin{eqnarray}
	\mathbf{r}_{t+1}=A_{r}(\theta_{0})\mathbf{z}_{t+1},\, m_{t+1}^{*}\boldsymbol{1}_{n_{m}}=A_{m^{*}}(\theta_{0})\mathbf{z}_{t+1},\,
	m_{t+1}\boldsymbol{1}_{n_{m}}=A_{m}\mathbf{z}_{t+1},\,\text{and}\, v_{t+1}=A_{v}(\theta_{0})\mathbf{z}_{t+1},\label{eqn:EndoVarSelection}
\end{eqnarray}
where $A_{r}(\theta_{0})$, $A_{m^{*}}(\theta_{0})$, and $A_{m}$ are $n_{m}\times n_{z}$ matrices and $A_{v}(\theta_{0})$ is a $1\times n_{z}$ vector that selects the appropriate elements from the set of endogenous variables.

\subsection{Irrelevance Theorem\label{subsec:Main-Theorem}}

We first characterize the equilibrium in terms of expectation and entropy terms.
\begin{prop}
\label{prop:all_plugged_in-1}If the model satisfies Assumption \ref{def:cond_affine},
Equation (\ref{eq:main_aff_tech}) can be written as
\begin{eqnarray}
b^{1}(\theta_{0})+B_{z}^{1}(\theta_{0})\mathbf{z}_{t}+B_{zl}^{1}(\theta_{0})\mathbf{z}_{t-1}+B_{x}^{1}(\theta_{0})\mathbf{x}_{t} & = & \mathbb{E}_{t}\left[D_{z}^{1}(\theta_{0})\mathbf{z}_{t+1}+D_{x}^{1}(\theta_{0})\mathbf{x}_{t+1}\right]\nonumber \\
	&  & +\mathbb{H}_{t}\left[D_{z}^{1}(\theta_{0})\mathbf{z}_{t+1}+D_{x}^{1}(\theta_{0})\mathbf{x}_{t+1}\right],\ \ \ \label{eq:main_aff_tech-1}
\end{eqnarray}
and Equation (\ref{eq:main_Euler}) can be written as
\begin{align}
\boldsymbol{0}_{n_{m}}= & \mathbb{E}_{t}[D_{z}^{2^{*}}(\theta_{0})\mathbf{z}_{t+1}]+\mathbb{H}_{t}\left[D_{z}^{2}(\theta_{0})\mathbf{z}_{t+1}\right]-\left(\frac{\frac{1}{\psi}-\gamma}{1-\gamma}\right)\mathbb{H}_{t}\left[(1-\gamma)A_{v}\text{\textbf{(\ensuremath{\theta_{0}})}}\mathbf{z}_{t+1}\right]\boldsymbol{1}_{n_{m}},\label{eq:main_Euler-1}
\end{align}
where $\boldsymbol{0}_{n_{m}}$ is an $n_{m}$-vector of zeros, $D_{z}^{2^{*}}(\theta_{0})\equiv A_{m^{*}}(\theta_{0})+A_{r}(\theta_{0})$, and $D_{z}^{2}(\theta_{0})\equiv A_{m}+A_{r}(\theta_{0})$. \end{prop}
\begin{proof}
Equation (\ref{eq:main_aff_tech-1}) is obtained from applying
Equation (\ref{eq:app_main_Hdecomp}) to the right-hand side of Equation
(\ref{eq:main_aff_tech}). Equation (\ref{eq:main_Euler-1}) results
from applying Equations (\ref{eq:main_sdf}), (\ref{eq:main_certainty_eq})
and (\ref{eqn:EndoVarSelection}) to the Euler Equation (\ref{eq:main_Euler}), and finally using Equation (\ref{eq:app_main_Hdecomp}) to obtain
the terms with the $\mathbb{H}_{t}$ operator.
\end{proof}

Proposition \ref{prop:all_plugged_in-1} divides equilibrium conditions into expectational terms and terms that depend on higher moments (those with $\mathbb{H}_{t}$). Expectational terms are independent of risk aversion and risk parameters $\theta_0$. Risk aversion only shows up in one of these higher moment terms: the last term on the right hand side of Equation (\ref{eq:main_Euler-1}), while the high-order parameters $\theta_h$ only show up in the entropy terms. This suggests that some properties of $\mathbf{z}_{t}$ can be determined without $\gamma$ or $\theta_{h}$. 

Our main theorem formalizes this intuition. It follows from conjecturing a solution for $\mathbf{z}_{t}$ and applying the method of undetermined coefficients. 
\begin{thm}
\label{prop:main_theorem}
Suppose Assumptions \ref{def:cond_affine} and \ref{def:cond_affine_CGF} hold. Then, the  endogenous variables of the model characterized by equilibrium conditions \eqref{eq:main_aff_tech-1} and \eqref{eq:main_Euler-1}in Proposition \ref{prop:all_plugged_in-1} follow the process
\begin{eqnarray}
\mathbf{z}_{t}=z+Z_{z}\mathbf{z}_{t-1}+Z_{x}\mathbf{x}_{t}+Z_{h}\mathbf{h}_{t},\label{eq:main_aff_solution-1}
\end{eqnarray}
where $z$ is an $n_{z}$-vector, $Z_{z}$ is an $n_{z}\times n_{z}$ matrix, $Z_{x}$ is an $n_{z}\times n_{x}$ matrix, and $Z_{h}$ is an $n_{z}\times n_{h}$ matrix, and 
\begin{center}
$Z_{z}$ and $Z_{x}$ do not depend on $\gamma$ or $\theta_{h}$ .
\end{center}
\end{thm}
The proof is in Appendix \ref{subsec:appendix_abstract}.

Theorem \ref{prop:main_theorem} says that risk aversion $\gamma$ and higher moment parameters $\theta_{h}$ are irrelevant for the elasticity of endogenous variables $\mathbf{z}_{t}$ with respect to first moment states $\text{\textbf{z}}_{t-1}$ and $\text{\textbf{x}}_{t}$.

\subsection{Extensions of the Irrelevance Theorem\label{subsec:TheoremExtension}}

Irrelevance applies to a broader set of preferences beyond EZ. In this section, we illustrate how to extend our framework to nest these and other preferences.

First, the set of model parameters is now defined more broadly as
\begin{align}
\Theta\equiv & \{\Lambda,\theta_{h},\theta_{0}\},\label{eq:general_Theta}
\end{align}
where $\Lambda$ is a \emph{vector }of generalized risk parameters that are not in $\theta_{0}$. By generalized risk, we mean $\Lambda$ includes parameters that control attitudes toward ambiguity or other attitudes not captured by Arrow-Pratt risk aversion. Second, the stochastic discount factor in Equation (\ref{eq:main_sdf}) is replaced with
\begin{equation}
m_{t+1}=m_{t+1}^{*}+\sum_{i=1}^{n_{i}}m_{i,t+1}.\label{eq:general_sdf}
\end{equation}
where $m_{t+1}^{*}$ is defined as before, and the risk adjustments $m_{i,t+1}$ are given by
\begin{equation}
m_{i,t+1}=\tilde{f}_{i}(\Lambda,\theta_{0})A_{i}\text{\textbf{z}}_{t+1},\label{eq:general_sdf-1}
\end{equation}
for all $i\in\{1,...,n_{i}\},$ and $A_{i}$ is a $1\times n_{i}$ selection vector to obtain the relevant state variables from $\text{\textbf{z}}_{t+1}$. Third, the certainty equivalent Equation (\ref{eq:main_certainty_eq}) is replaced with
\begin{equation}
\mathbb{E}_{i,t}[A_{i}\text{\textbf{z}}_{t+1}]=-\mathbb{H}_{i,t}[\hat{f}_{i}(\Lambda,\theta_{0})A_{i}\text{\textbf{z}}_{t+1}],\label{eq:general_sdf-1-1}
\end{equation}
where the subscripts $i$ on the expectation and entropy terms indicate that these terms may use different probability measures (e.g., subjective beliefs), and may not necessarily use the objective measure.

The extensions (\ref{eq:general_Theta})-(\ref{eq:general_sdf-1-1}) maintain the key property that leads to irrelevance. As described in Section \ref{subsec:ez_Irrelevance}, the key property is that the risk adjustments are solely functions of the higher-moment variables $\text{\textbf{h}}_{t}$. To see this, the extension implies that Equation (\ref{eq:main_Euler-1}) in Proposition \ref{prop:all_plugged_in-1} is generalized as
\begin{align}
\boldsymbol{0}_{n_{m}}= & \mathbb{E}_{t}[D_{z}^{2^{*}}(\theta_{0})\mathbf{z}_{t+1}]+\mathbb{H}_{t}\left[D_{z}^{2}(\theta_{0})\mathbf{z}_{t+1}\right]-\sum_{i=1}^{n_{i}}\left(\frac{\tilde{f}_{i}(\Lambda,\theta_{0})}{\hat{f}_{i}(\Lambda,\theta_{0})}\right)\mathbb{H}_{i,t}\left[\hat{f}_{i}(\Lambda,\theta_{0})A_{i}\mathbf{z}_{t+1}\right]\boldsymbol{1}_{n_{m}},\label{eq:main_Euler-1-2}
\end{align}
where $D_{z}^{2}(\theta_{0})$ is now defined as $D_{z}^{2}(\theta_{0})\equiv A_{m^{*}}(\theta_{0})+\sum_{i=1}^{n_{i}}A_{i}+A_{r}(\theta_{0})$. Using this generalization, the logic of Theorem \ref{prop:main_theorem} remains valid.

To illustrate this extension we present a couple of examples:
\begin{example}
(Recursive Preferences) It is easy to show that the recursive preference
specification corresponds to the particular case $\Lambda=\gamma$,
$\tilde{f}_{1}(\Lambda,\theta_{0})=\frac{1}{\psi}-\gamma$,$\hat{f}_{1}(\Lambda,\theta_{0})=1-\gamma$,
and $A_{1}\text{\textbf{z}}_{t+1}=v_{t+1}-w_{t}$, where expectations
in Equation (\ref{eq:general_sdf-1-1}) are expectations under the
objective measure.
\end{example}

\begin{example}
(Generalized Smooth Ambiguity) As in \citet{ju2012ambiguity}, the
stochastic discount factor is
\begin{equation}
m_{t+1}=m_{t+1}^{*}+\left(\frac{1}{\psi}-\gamma\right)(v_{t+1}-w_{\pi_{s},t})+\left(\frac{1}{\psi}-\eta\right)(w_{\pi_{s},t}-w_{t}),\label{eq:main_sdf-1}
\end{equation}
where $\eta$ is the coefficient of ambiguity aversion,
\begin{align}
w_{t} & =\frac{1}{1-\eta}\log\mathbb{E}_{\mu,t}\left[\exp((1-\eta)w_{\pi_{s},t})\right],\label{eq:main_certainty_eq-1}
\end{align}
and
\begin{align}
w_{\pi_{s},t} & =\frac{1}{1-\gamma}\log\mathbb{E}_{\pi_{s},t}\left[\exp((1-\gamma)v_{t+1})\right].\label{eq:main_certainty_eq-1-1}
\end{align}
The expectation in Equation (\ref{eq:main_certainty_eq-1}) is obtained under the belief system $\mu$ of the representative agent over a set of probability distributions denoted by $\pi_{s}$ for a hidden state $s$, while the expectation in Equation (\ref{eq:main_certainty_eq-1-1}) is under the distribution $\pi_{s}$.

Equations (\ref{eq:main_sdf-1})-(\ref{eq:main_certainty_eq-1-1}) can be written as Equations (\ref{eq:general_sdf})-(\ref{eq:general_sdf-1-1}) where $\Lambda=\{\gamma,\eta\}$, $\tilde{f}_{1}(\Lambda,\theta_{0})=\frac{1}{\psi}-\gamma$, $\hat{f}_{1}(\Lambda,\theta_{0})=1-\gamma$, $A_{1}\text{\textbf{z}}_{t+1}=v_{t+1}-w_{\pi_{s},t}$, $\tilde{f}_{2}(\Lambda,\theta_{0})=\frac{1}{\psi}-\eta$, $\hat{f}_{2}(\Lambda,\theta_{0})=1-\eta$, and $A_{2}\text{\textbf{z}}_{t+1}=w_{\pi_{s},t}-w_{t}$. Irrelevance holds if belief dynamics are exogenous, as in \citet{ilut2014ambiguous}. \end{example}

\section{Irrelevance and Global Solution Methods\label{sec:Numerical-Verification}}

Approximate linearity of constraints is a key assumption of our theorem. The framework also assumes that the CGF of shocks is conditionally affine. One may wonder about the fragility of these assumptions. To address this question, this section examines irrelevance in two classes of models that are solved numerically using projection methods (\citealt{judd1992projection}). Both classes build on the model in Section \ref{sec:tallarini_extension}.

The first class of models builds on a fixed-labor version of \citet{tallarini2000risk}.\footnote{\citet{tallarini2000risk} uses a second-order Taylor expansion solution
(based on \citet{christiano1990solving}) which implies constant risk
premiums. The projection method allows us to capture time variation
in risk premiums.} The second class of models introduces additional non-linearities in the model structure as well as skewness and fat-tailed shocks. The additional non-linearity comes from lower EIS and larger capital adjustment costs, which have been shown to produce more volatile equity returns and a larger equity premium (see, for instance, \citet{jermann1998asset}, \citet{guvenen2009parsimonious}, \citet*{kaltenbrunner2010long}, \citet{papanikolaou2011investment}, \citet{chen2017external}). The skewed and fat-tailed shocks come in the form of gamma-distributed productivity shocks, calibrated to be consistent with the empirical evidence of fat tails in \citet*{bekaert2017asset}.

\subsection{Additional Model Structure and Common Parameters\label{subsec:Additional-Structure-and}}

To allow for skewed and fat-tailed shocks, we generalize the model in Section \ref{sec:tallarini_extension} such that $\epsilon_{x,t+1}$ follows a zero-mean left-skewed gamma distribution
\begin{align}
\epsilon_{x,t+1}= & \zeta\theta_{t}-\text{gamrv}(\zeta,\theta_{t}),\label{eq:num_tvgamma1}
\end{align}
where $\zeta$ is the gamma shape parameter, and $\theta_{t}$ is the time-varying gamma scale parameter. This parameter is linked to the higher moment state $h_{t}$ according to

\begin{align}
\theta_{t}\equiv & \frac{\sigma_{x}}{\sqrt{\zeta}}\exp(h_{t}).\label{eq:num_tvgamma2}
\end{align}
Using properties of the gamma distribution, the conditional standard deviation of the productivity shock is $\sigma_{x}\exp(h_{t}),$ and thus $h_{t}$ can be interpreted as the log-deviation of volatility from its steady state.

Equations (\ref{eq:num_tvgamma1}) and (\ref{eq:num_tvgamma2}) imply that the skewness and kurtosis of productivity shocks are
\begin{align}
\text{Skew}_{t}(\epsilon_{x,t+1})= & -\frac{2}{\sqrt{\zeta}},\quad\text{{and} \  }\text{Kurt}_{t}(\epsilon_{x,t+1})=3+\frac{6}{\zeta},\label{eq:struct_skew_kurt}
\end{align}
respectively, implying left-skewness and fat tails. It can be shown that, as $\zeta\longrightarrow\infty$, skewness goes to zero and kurtosis goes to 3, leading to a normal distribution.

Panel A of Table \ref{tab:param} lists parameters that are common across specifications for the two classes of models we analyze, which facilitates comparisons across models. Most of these parameter values are standard in the literature. For example, we choose the volatility of productivity to imply that HP-filtered output has a volatility of about 1.5\% per quarter, close to values estimated from data. The baseline risk aversion $\gamma=1$ and the baseline volatility of volatility $\sigma_{h}=0$ correspond to log utility and homoskedastic shocks, respectively. We vary these two parameters to analyze models in Sections \ref{subsec:Class-1:-Frictionless} and \ref{subsec:Class-2:-Large}, but use these baseline values when the given parameter is not the focus.

\begin{table}[th]
\caption{Parameter Values for the Numerical Exercises}
\label{tab:param}

\begin{singlespace}
\noindent All exercises use the equations for the model in Section
\ref{sec:tallarini_extension} with the additional structure described in Section
\ref{subsec:Additional-Structure-and}. The model parameters are quarterly.
The baseline value for risk aversion is the value used when volatility
of volatility varies in Figures \ref{fig:tallarini} and \ref{fig:intermediate},
and conversely the baseline value of volatility of volatility is used
when risk aversion varies.
\end{singlespace}

\begin{centering}
\vspace{4ex}
\par\end{centering}
\begin{centering}
\begin{tabular}{lcc}
\toprule
\multicolumn{3}{c}{Panel A: Common Parameters}\tabularnewline
\midrule
Time Preference & $\beta$ & 0.99\tabularnewline
Capital Share & $\alpha$ & 0.35\tabularnewline
Depreciation Rate & $\delta$ & 0.02\tabularnewline
 &  & \tabularnewline
Volatility of Productivity & $\sigma_{x}$ & 0.014\tabularnewline
Persistence of Productivity & $\phi_{x}$ & 0.95\tabularnewline
Persistence of Volatility & $\phi_{h}$ & 0.84\tabularnewline
 &  & \tabularnewline
Baseline Risk Aversion & $\gamma$ & 1\tabularnewline
Baseline Volatility of Volatility & $\sigma_{h}$ & 0\tabularnewline
\bottomrule
\end{tabular}
\par\end{centering} \vspace{4ex}

\centering{}%
\begin{tabular}{lcccc}
\toprule
\multicolumn{5}{c}{Panel B: Parameters That Vary Across Classes of Models}\tabularnewline
\midrule
 &  & Class 1 &  & Class 2\tabularnewline
Elasticity of Intertemporal Substitution & $\psi$ & 1 &  & 0.30\tabularnewline
Elasticity of Investment & $\xi$ & 10 &  & 3\tabularnewline
Gamma Shape & $\zeta$ & 600 &  & 6\tabularnewline
\bottomrule
\end{tabular}
\end{table}

The persistence of volatility does not have a standard value in the literature. Some papers find volatility to be nearly a random walk (e.g., \citet*{bansal2004risks}, \citet*{bansal2012empirical}), while others imply high but far-from random walk persistence (e.g., \citet*{bansal2005interpretable}, \citet{bloom2009impact}, \citet{chen2017external}). Our chosen persistence of 0.84 is helpful to ensure that our numerical solution both converges and is accurate, as near-random walk persistence is known to lead to large approximation errors as shown by \citet*{pohl2018higher}.

Panel B of Table \ref{tab:param} describes the parameters that differ across our two classes of models. We discuss these parameters further in Sections \ref{subsec:Class-1:-Frictionless} and \ref{subsec:Class-2:-Large}.

\subsection{Projection Solution Method and Model Simulation}

We solve the model using projection methods as in \citet{judd1992projection} and \citet*{caldara2012computing}. As shown by \citet*{caldara2012computing}, projection methods offer the highest accuracy of a number of solution methods. This method can be computationally slow and it can be difficult to find an initial guess that leads to convergence. To deal with this limitation, we use a homotopy method as in \citet{chen2017external}.

We begin by discretizing the volatility process using \citet{rouwenhorst1995asset}, followed by discretizing productivity using \citet{tauchen1986finite}. The Rouwenhorst method has been shown to be the most accurate method for approximating AR(1) processes (\citet{kopecky2010finite}), but it cannot accommodate time-varying volatility, leading us to use Tauchen for productivity. We then use cubic splines to approximate the law of motion for capital. Breakpoints are linearly spaced for (log) capital, with standard knot-averaging nodes.

To solve the model, we use Broyden's method to find cubic spline coefficients that satisfy both the firm's Euler equation and the household's value recursion. The model is considered solved if the maximum error on the collocation nodes is less than $1.48\times10^{-8}.$

As our two classes of models differ strongly in their non-linearity and non-normality, we choose the degree of numerical approximation (number of gridpoints) to match the task at hand. For models with normal shocks we use 25 gridpoints for productivity while for models with skewed shocks we use 71 gridpoints for productivity. Appendix \ref{subsec:Nx_demo} shows that such a large amount of gridpoints is necessary for the Tauchen approximation to produce moments that are similar to those of the continuous gamma distribution. To accommodate the large number of gridpoints for productivity, both models use 5 gridpoints for volatility and 8 gridpoints for capital.\footnote{Our solution makes heavy use of the \citet{miranda2001applied} COMPECON
toolbox. Code for the solution is available at https://sites.google.com/site/chenandrewy/.}

To study irrelevance, we simulate model dynamics under a wide range of values for risk aversion and the volatility of volatility. In each quarter of each simulation, we calculate the following variables of interest: the conditional equity premium, the elasticity of the stock price with respect to productivity ($\partial s_{t}/\partial x_{t}$), and the elasticity of the stock price with respect to volatility ($\partial s_{t}/\partial h_{t}$). Finally, taking the 25th, 50th, and 75th percentiles within each set of simulations provides a sense of the typical range of values for each variable. 

\subsection{Class 1: Frictionless Models with Normal Shocks (a la Tallarini 2000)\label{subsec:Class-1:-Frictionless}}

We begin by examining a class of models in the spirit of \citet{tallarini2000risk}, using the model parameterization in Table \ref{tab:param}. Given the EIS of 1, the elasticity of investment of 10 is chosen to lead to a relative volatility of consumption of about 0.50\%. Similarly, the gamma shape parameter of 600 leads to normally-distributed productivity shocks, as in Tallarini (see Equation (\ref{eq:struct_skew_kurt})).

The simulation results are shown in Figure \ref{fig:tallarini}. The top row illustrates how our endogenous variables of interest vary as the degree of risk aversion changes. For this analysis, we hold the volatility of volatility constant at 0 (Table \ref{tab:param}, Panel B). In line with Theorem \ref{prop:main_theorem}, risk aversion has essentially no effect on the elasticity of the stock price with respect to productivity $\partial s_{t}/\partial x_{t}$ (middle panel, top row). The median elasticity declines slightly as risk aversion increases, but this variation is negligible relative to the level of this elasticity. Notably, this irrelevance holds even if risk aversion is 100.

\begin{figure}[th]
    \centering{}\caption{\textbf{\label{fig:tallarini}{\color{ChadBlue}Numerical Check for Irrelevance - Model Class 1: Frictionless Investment and Conditionally Normal  Shocks (a la Tallarini 2000).}}
    Models' equations are described in Sections \ref{sec:tallarini_extension} and
    \ref{subsec:Additional-Structure-and}. Parameter values that are
    held fixed are described in Table \ref{tab:param}. We vary risk aversion
    $\gamma$ and the conditional volatility of volatility $\sigma_{h}$,
    and simulate the model, calculating variables of interest in each
    quarter of each simulation. The charts show the 25th, 50th, and 75th
    percentiles within each simulation. Vol of vol is plotted in terms
    of its unconditional value implied by its AR(1) process. Models are
    solved using projection methods.}
    \includegraphics[width=1\textwidth]{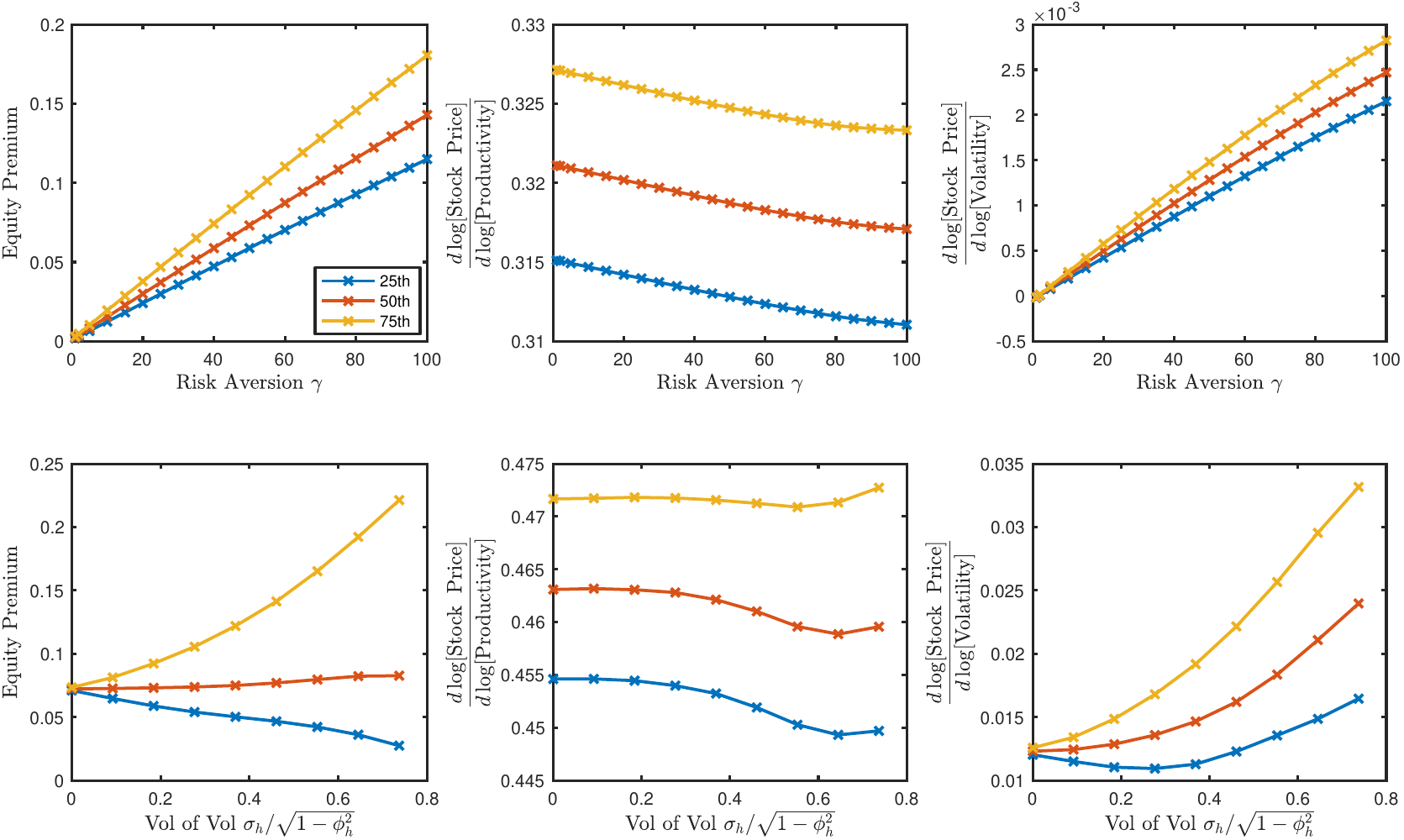}
\end{figure}

In contrast, risk aversion has notable effects on the equity premium and the elasticity of the stock price with respect to volatility $\partial s_{t}/\partial h_{t}$ (left and right panels, top row). Unlike $\partial s_{t}/\partial x_{t}$ , these two variables are not subject to Theorem \ref{prop:main_theorem}, which only applies to the elasticity of variables with respect to first moment states like productivity. Indeed, the Theorem suggests that both the equity premium and $\partial s_{t}/\partial h_{t}$ do, in principle, depend on risk aversion. Whether or not a notable effect is seen in a numerical solution, however, depends on the details of the model.

The bottom panels of Figure \ref{fig:tallarini} illustrate how aspects of the model solution change as time-varying risk changes. Risk aversion in these panels is held constant at 1 (Table \ref{tab:param} Panel B). The results are again consistent with Theorem \ref{prop:main_theorem}. Time-varying volatility is irrelevant for $\partial s_{t}/\partial x_{t}$. The median elasticity declines slightly as vol of vol increases, but the decrease is negligible, even if the unconditional volatility of volatility is 80\%. Conversely, the equity premium and $\partial s_{t}/\partial h_{t}$ change as vol of vol increases.

Overall, these results clarify and extend Tallarini's (2000) seminal paper. Tallarini's irrelevance result is often cited as a separation between macroeconomics and finance variables. The top-middle panel of Figure \ref{fig:tallarini} shows, however, that risk aversion is also irrelevant for stock price volatility in Tallarini's setting. Rather, the separation is between means (like the equity premium) and volatilities (like stock price volatility), as long as volatility is held constant, consistent with our theoretical results.

\subsection{Class 2: Large Investment Frictions and Gamma Productivity Shocks
\label{subsec:Class-2:-Large}}

The models studied in Section \ref{subsec:Class-1:-Frictionless} feature extreme risk aversion or elevated volatility of volatility, but they still have a limited role for risk in the sense that the equity premium is small as well as the elasticity of stock prices to productivity. Moreover, the normally distributed shocks in this class of models implies that our affine approximation of the CGF holds exactly. Here, we study a class of models that produces a more economically significant risk premium and has skewed, fat-tailed shocks. We examine models with an inelastic investment parameter $\xi$ of 3, a low EIS of 0.3, and a gamma shape parameter $\zeta$ of 6 (Table \ref{tab:param}, Panel B), which implies a significant deviation from a normal distribution.

Large capital adjustment costs help produce a large equity premium as shown in \citet{jermann1998asset}, \citet{guvenen2009parsimonious}, \citet*{kaltenbrunner2010long}, and \citet{chen2017external}, among others. This result is a natural consequence of Q-theory, which shows that stock price volatility increases with the adjustment cost, all else being equal. High adjustment costs, however, tend to lead to excessively smooth investment and excessively volatile consumption. Thus, the aforementioned models feature a low EIS, which helps keep consumption smooth. 

Our choice of an EIS of 0.30 matches the EIS of stockholders in \citet{guvenen2009parsimonious}. Given this EIS, we find that in our model an elasticity of investment of 3 leads to a relative volatility of consumption of about 0.50\%.\footnote{Guvenen's model has a much smaller elasticity of investment of 0.40,
but his model is a richer one with heterogeneous agents. Indeed, the
stockholders in his model have more than double the consumption volatility
of non-stockholders.} Lower values for the EIS can be justified using the classical consumption Euler equation regressions of \citet{hall1988intertemporal}, but such values are rarely found in macroeconomic models.

Our gamma shape parameter of 6 is chosen to be consistent with the empirical evidence in \citet*{bekaert2017asset}. The shape parameter of 6 implies a kurtosis of 4 and a skewness of -0.82, similar to Bekaert and Engstrom's finding that quarterly consumption growth has a kurtosis of 4.04 and a skewness of -0.399.

As shown in Figure \ref{fig:intermediate}, this class of models produces an equity premium that is more than 3 times larger than that implied by our first class (top left panel). Similarly, the elasticity of stock prices with respect to productivity $\partial s_{t}/\partial x_{t}$ is about 3 times larger (top middle). Despite the much larger risk premium, this class of models continues to display irrelevance: risk aversion and time-varying risk have essentially no effect on $\partial s_{t}/\partial x_{t}$ (middle panels). In contrast, risk aversion and time-varying risk do have an effect on average values like the equity premium (left panels) or on how the model responds to volatility (right panels).

\begin{figure}[th]
    \centering{}\caption{\textbf{\label{fig:intermediate}{\color{ChadBlue}Numerical Check for Irrelevance - Model Class 2: Large Investment Frictions and Negatively-Skewed Gamma Shocks.}}
    Models' equations are described in Sections \ref{sec:tallarini_extension} and
    \ref{subsec:Additional-Structure-and}. Parameter values that are
    held fixed are described in Table \ref{tab:param}. We vary risk aversion
    $\gamma$ and the conditional volatility of volatility $\sigma_{h}$,
    and simulate the model, calculating variables of interest in each
    quarter of each simulation. The charts show the 25th, 50th, and 75th
    percentiles within each simulation. Vol of vol is plotted in terms
    of its unconditional value implied by its AR(1) process. Models are
    solved using projection methods.}
    \includegraphics[width=1\textwidth]{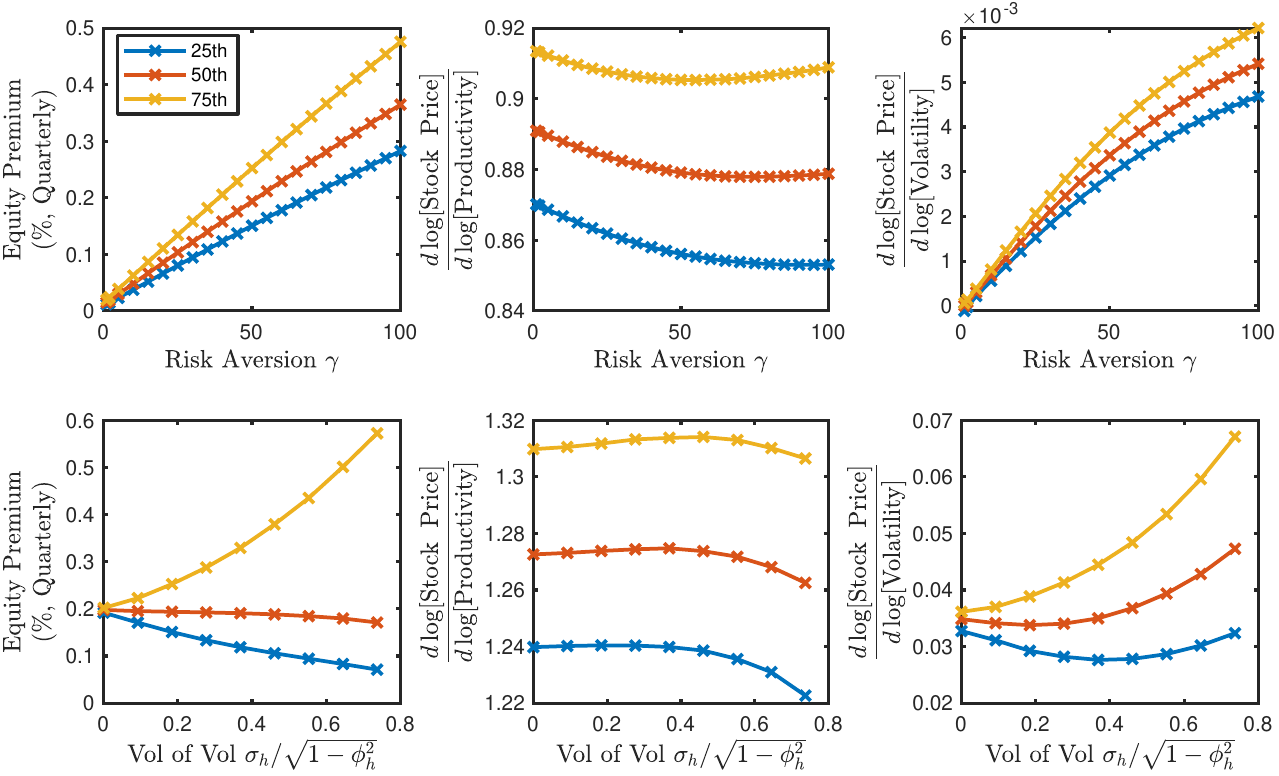}
\end{figure}

This exercise shows that irrelevance also holds in a model with non-linearities and fat tails that are consistent with empirical evidence. These results suggest that the irrelevance result obtained using an affine approximation is also displayed in numerical solutions. This finding is consistent with findings in a wide variety of models in the literature that exhibit irrelevance such as \citet*{rudebusch2012bond}, \citet{croce2014investor}, and \citet*{petrosky2018endogenous}. 

However, as highlighted in Section \ref{sec:ezmodel}, this result does not extend to models with strong non-linearities such as \cite{brunnermeier2014macroeconomic} and \cite{GourioNgo2020ZLB} for which affine approximations are inaccurate. A common feature of these models is occasionally binding constraints, which is missing from Figures \ref{fig:tallarini} and \ref{fig:intermediate}. 

\clearpage
\section{Conclusion\label{sec:Conclusion}}

\noindent We provide an irrelevance theorem that sheds light on the role of risk in business cycle modeling. If a model exhibits (1) separation of intertemporal and risk preferences, (2) separation of the drivers of first and higher moments in the remaining economic structure, and (3) approximate linearity of constraints, then risk is irrelevant for the elasticity of any variable with respect to first moment states. We prove the theorem by separating the equilibrium conditions into expectational terms and high-order terms, and showing that risk aversion and risk parameters are absent in the expectational equations and then only drive high-order terms. Irrelevance holds for any variable, including asset prices, clarifying and extending the \citet{tallarini2000risk} separation result. 

The theorem also provides a benchmark for understanding how economic mechanisms in salient models in the literature give a more meaningful role to risk and risk attitudes in business cycle modeling. One can either ``break'' irrelevance with economic mechanisms that remove one of the three key properties above. Or ``adapt'' to it, by driving business cycles with shocks to higher moments instead of first moments. 

Either way, our theorem suggests that meaningful risk modeling in production economies is nontrivial. One must either model interactions between intertemporal and risk tradeoffs, incorporate strong nonlinearities, or adapt models to incorporate mechanisms or frictions that avoid the \citet{barro1984time} comovement result.

\newpage{}

\appendix


\section{Derivations for the Dynamic Model (Section \ref{sec:tallarini_extension}) \label{subsec:Appendix_ezmodel}}

\subsection{Equilibrium Conditions}

First let's derive the equilibrium conditions. The optimality of firm investment implies marginal Q and the investment return are
\begin{align*}
Q_{t}= & \frac{1}{\Phi'(I_{t}/K_{t})},\\
R_{i,t+1}= & Q_{t}^{-1}\left\{ \alpha X_{t+1}K_{t+1}^{-(1-\alpha)}+Q_{t+1}\left[\Phi(I_{t+1}/K_{t+1})+1-\delta\right]-\frac{I_{t+1}}{K_{t+1}}\right\} ,
\end{align*}
respectively, where the investment return satisfies the Euler Equation
\begin{align*}
1= & \mathbb{E}_{t}\left(M_{t+1}R_{i,t+1}\right),\\
M_{t+1}= & \beta\left(\frac{C_{t+1}}{C_{t}}\right)^{-\frac{1}{\psi}}\left(\frac{V_{t+1}}{P_{t}}\right)^{\frac{1}{\psi}-\gamma}.
\end{align*}
Since production is homogeneous of degree 1, Tobin's Q theory holds, and then the ex-dividend stock price is
\begin{align*}
S_{t}= & Q_{t}K_{t}.
\end{align*}

Collecting our equilibrium conditions together, we have
\begin{align}
C_{t}+I_{t}= & X_{t}K_{t}^{\alpha}\label{eq:app_ez_rc}\\
K_{t+1}= & (1-\delta)K_{t}+\Phi(I_{t}/K_{t})K_{t}\label{eq:app_ez_cap}\\
Q_{t}= & \frac{1}{\Phi'(I_{t}/K_{t})}\label{eq:app_ez_Q}\\
R_{i,t+1}= & Q_{t}^{-1}\left\{ \alpha X_{t+1}K_{t+1}^{-(1-\alpha)}+Q_{t+1}\left[\Phi(I_{t+1}/K_{t+1})+1-\delta\right]-\frac{I_{t+1}}{K_{t+1}}\right\} \label{eq:app_ez_rinv}\\
V_{t}= & \left\{ (1-\beta)C_{t}^{1-1/\psi}+\beta P_{t}{}^{1-1/\psi}\right\} ^{1/(1-1/\psi)}\label{eq:app_ez_v_recur}\\
P_{t}= & \left[\mathbb{E}_{t}(V_{t+1}^{1-\gamma})\right]^{1/(1-\gamma)}\label{eq:app_ez_ce}\\
M_{t+1}= & \beta\left(\frac{C_{t+1}}{C_{t}}\right)^{-\frac{1}{\psi}}\left(\frac{V_{t+1}}{P_{t}}\right)^{\frac{1}{\psi}-\gamma}.\label{eq:app_ez_SDF}\\
1= & \mathbb{E}_{t}\left(M_{t+1}R_{i,t+1}\right)\label{eq:app_ez_Euler}\\
S_{t}= & Q_{t}K_{t}.\label{eq:app_ez_stock}
\end{align}

\subsection{Affine Approximation\label{subsec:app_ezmodel_Affine-Approximation}}

We express all variables in terms of log deviations from an approximation point. For example, for productivity $X_{t}$ we define an approximation point $\bar{X}$ and a log deviation $x_{t}$ so that $X_{t}=\bar{X}e^{x_{t}}$. This approximation point is not necessarily the non-stochastic steady state. Nevertheless, we choose the approximation points strategically, so that the algebra simplifies.

Also, for convenience, we define
\begin{align*}
\overline{CY}\equiv & \frac{\bar{C}}{\bar{Y}},
\end{align*}
and other similar variables $\overline{IY}$, $\overline{YK}$, etc. are defined similarly.

First, we deal with the expectational equations. We do not approximate these equations in order to preserve risk terms in the model. Nevertheless, we rewrite these equations in terms of log-deviations from an approximation point to make the form match the approximated equations. Specifically, the certainty equivalent (\ref{eq:app_ez_ce}) becomes
\begin{align*}
\bar{W}^{(1-\gamma)}e^{(1-\gamma)w_{t}}= & \bar{V}^{(1-\gamma)}\mathbb{E}_{t}e^{(1-\gamma)v_{t+1}}\\
\Rightarrow(1-\gamma)w_{t}=\log & \left(\frac{\bar{W}}{\bar{V}}\right)^{-(1-\gamma)}+\log\mathbb{E}_{t}\exp\left[(1-\gamma)v_{t+1}\right]\\
\Rightarrow w_{t}=-\log & \left(\frac{\bar{W}}{\bar{V}}\right)+\frac{1}{(1-\gamma)}\log\mathbb{E}_{t}\exp\left[(1-\gamma)v_{t+1}\right].
\end{align*}
Note that in an approximation around the non-stochastic steady state, $\bar{W}=\bar{V}$ and the intercept term disappears, but we keep this intercept to be clear about the generality of our results.

Similarly the stochastic discount factor (\ref{eq:app_ez_SDF}) is simply rearranged rather than approximated
\begin{align*}
\bar{M}e^{m_{t+1}}= & \beta\frac{\bar{C}}{\bar{C}}e^{-\psi^{-1}\Delta c_{t+1}}\frac{\bar{V}}{\bar{W}}e^{(\psi^{-1}-\gamma)(v_{t+1}-w_{t})}\\
\Rightarrow m_{t+1}= & \log\left(\frac{\beta\bar{V}}{\bar{M}\bar{W}}\right)-\psi^{-1}\Delta c_{t+1}+(\psi^{-1}-\gamma)(v_{t+1}-w_{t}),
\end{align*}
and the Euler equation is rearranged to be
\begin{align*}
0= & \log(\bar{M}\bar{R})+\log\mathbb{E}_{t}\exp\left(m_{t+1}+r_{t+1}\right).
\end{align*}

Second, we combine the technological constraints into a single linear expression. To do this, begin by applying Uhlig's approach to the resource constraint (\ref{eq:app_ez_rc})
\begin{align}
\bar{C}(1+c_{t})+\bar{I}(1+i_{t})= & \bar{Y}(1+x_{t}+\alpha k_{t}),\nonumber \\
\overline{CY}c_{t}+\overline{IY}i_{t}= & x_{t}+\alpha k_{t}.\label{eq:app_ez_cap2}
\end{align}
The tricky part is the linearization of the adjustment cost and capital accumulation equations. To linearize these equations, first we use a Taylor expansion of the adjustment cost function around the steady state investment rate $\delta$:
\begin{align*}
\Phi(I_{t}/K_{t})\approx & \bar{\Phi}+\Phi'(\delta)(I_{t}/K_{t}-\delta)\\
\Rightarrow\Phi_{t}= & \bar{\Phi}+I_{t}/K_{t}-\delta,
\end{align*}
where the second line defines $\Phi_{t}=\Phi(I_{t}/K_{t})$ and the fact that $\Phi'(\delta)=1$. Then, we apply Uhlig's log-linearization to $\Phi_{t},$ $I_{t}$, and $K_{t}$:
\begin{align}
\bar{\Phi}(1+\phi_{t})= & \bar{\Phi}+\delta(1+i_{t}-k_{t})-\delta\nonumber \\
\Rightarrow\phi_{t} & =i_{t}-k_{t}.\label{eq:app_ez_phi}
\end{align}
Now to finish up, we apply Uhlig's approach to capital accumulation:
\begin{align*}
\bar{K}(1+k_{t+1})= & (1-\delta)\bar{K}(1+k_{t})+\bar{\Phi}\bar{K}(1+\phi_{t}+k_{t}),
\end{align*}
and plug in our simple expression for $\phi_{t}$ to obtain
\begin{align*}
k_{t+1}= & (1-\delta)k_{t}+\delta i_{t}.
\end{align*}
Finally, plugging in the above equation into the linearized resource constraint (\ref{eq:app_ez_cap2}) gives the key technological constraint of the model
\begin{align}
\overline{CY}c_{t}= & x_{t}+\left[\alpha+\bar{KY}(1-\delta)\right]k_{t}-\bar{KY}k_{t+1}.\label{eq:app_ez_cap3}
\end{align}

Third, we linearize the marginal rates of transformation. The tricky part here is the linearization of Tobin's Q in the presence of the adjustment cost. To do this, we first use a Taylor expansion of the first derivative of the adjustment cost
\begin{align*}
\Phi'\left(I_{t}/K_{t}\right)\approx & \Phi'(\delta)+\Phi''(\delta)(I_{t}/K_{t}-\delta)\\
= & 1-\frac{\xi^{-1}}{\delta}(I_{t}/K_{t}-\delta).
\end{align*}
We plug the above into the marginal Q equation, and then apply Uhlig's log-linearizations:
\begin{align*}
Q_{t}^{-1}= & \Phi'(I_{t}/K_{t})\\
\Rightarrow\bar{Q}^{-1}(1-q_{t}) & =1-\frac{\xi^{-1}}{\delta}(I_{t}/K_{t}-\delta)\\
\Rightarrow q_{t}= & \xi^{-1}(i_{t}-k_{t}-\delta).
\end{align*}

Now we linearize the investment return (\ref{eq:app_ez_rinv}). We begin by moving $Q_{t}^{-1}$ to the LHS, and then apply Uhlig's approach, beginning by multiplying everything out
\begin{align*}
R_{i,t+1}Q_{t}= & \left\{ \alpha X_{t+1}K_{t+1}^{-(1-\alpha)}+Q_{t+1}\left[\Phi_{t+1}+1-\delta\right]-\frac{I_{t+1}}{K_{t+1}}\right\} \\
= & \alpha X_{t+1}K_{t+1}^{-(1-\alpha)}+Q_{t+1}\Phi_{t+1}+Q_{t+1}(1-\delta)-\frac{I_{t+1}}{K_{t+1}},
\end{align*}
and then apply Uhlig's rules for log-linearization
\begin{align*}
\bar{R}\bar{Q}(1+r_{i,t+1}+q_{t})= & \alpha\overline{YK}[1+x_{t+1}-(1-\alpha)k_{t+1}]\\
 & +\bar{Q}\bar{\Phi}(1+q_{t+1}+\phi_{t+1})\\
 & +\bar{Q}(1-\delta)(1+q_{t+1})-\delta(1+i_{t+1}-k_{t+1})\\
= & \alpha\overline{YK}[1+x_{t+1}-(1-\alpha)k_{t+1}]\\
 & +\delta(1+q_{t+1}+i_{t}-k_{t})\\
 & +(1-\delta)(1+q_{t+1})-\delta(1+i_{t+1}-k_{t+1})\\
= & \alpha\overline{YK}[1+x_{t+1}-(1-\alpha)k_{t+1}]+(1+q_{t+1}),
\end{align*}
where the second equality uses Equations (\ref{eq:app_ez_phi}) and the fact that $\Phi(\delta)=\delta$. Finally, we solve for $r_{i,t+1}$
\begin{align*}
\bar{R}r_{i,t+1}= & (\alpha\overline{YK}+1-\bar{R})\\
 & +\alpha\overline{YK}(x_{t+1}-(1-\alpha)k_{t+1})+q_{t+1}-\bar{R}q_{t}.
\end{align*}

To log-linearize the value recursion, we simply rearrange Equation (\ref{eq:app_ez_v_recur}) and apply Uhlig's rules:
\begin{align}
V_{t}^{(1-\psi^{-1})}= & (1-\beta)C_{t}^{1-\psi^{-1}}+\beta P_{t}^{1-\psi^{-1}}\nonumber \\
\Rightarrow\bar{V}^{1-\psi^{-1}}(1+(1-\psi^{-1})v_{t})= & (1-\beta)\bar{C}^{1-\psi^{-1}}(1+(1-\psi^{-1})c_{t})+\beta\bar{W}^{1-\psi^{-1}}(1+(1-\psi^{-1})w_{t})\nonumber \\
\Rightarrow\bar{V}^{1-\psi^{-1}}v_{t}= & (1-\beta)\bar{C}^{1-\psi^{-1}}c_{t}+\beta\bar{W}^{1-\psi^{-1}}w_{t},\label{eq:app_ez_v_recur_lin}
\end{align}
where the 3rd equality imposes the fact that the value recursion (\ref{eq:app_ez_v_recur}) holds at the approximation point. Note that we do not impose that the approximation point is the non-stochastic steady state (as in, say, \citealt{uhlig2010easy}). Imposing the non-stochastic steady state would simplify Equation (\ref{eq:app_ez_v_recur_lin}), but we skip this simplification to emphasize the generality of our results.

Putting the approximate equilibrium conditions together, we summarize the system as
\begin{align}
\overline{CY}c_{t}+\overline{IY}i_{t}= & x_{t}+\alpha k_{t}\label{eq:app_ezlin_rc}\\
k_{t+1}= & (1-\delta)k_{t}+\delta i_{t}\\
q_{t}= & \xi^{-1}(i_{t}-k_{t}-\delta)\label{eq:app_ezlin_q}\\
\bar{R}r_{i,t+1}= & (\alpha\overline{YK}+1-\bar{R})+\alpha\overline{YK}(x_{t+1}-(1-\alpha)k_{t+1})+q_{t+1}-\bar{R}q_{t}\label{eq:app_ezlin_rinv}\\
\bar{V}^{1-\psi^{-1}}v_{t}= & (1-\beta)\bar{C}^{1-\psi^{-1}}c_{t}+\beta\bar{W}^{1-\psi^{-1}}w_{t}\label{eq:app_ezlin_v_recur}\\
w_{t}=\log & \left(\frac{\bar{W}}{\bar{V}}\right)+\frac{1}{(1-\gamma)}\log\mathbb{E}_{t}\exp\left[(1-\gamma)v_{t+1}\right]\label{eq:app_ezlin_ce}\\
m_{t+1}= & \log\left(\frac{\beta\bar{V}}{\bar{M}\bar{W}}\right)-\psi^{-1}\Delta c_{t+1}+(\psi^{-1}-\gamma)(v_{t+1}-w_{t})\\
0= & \log(\bar{M}\bar{R})+\log\mathbb{E}_{t}\exp\left(m_{t+1}+r_{t+1}\right)\label{eq:app_ezlin_Euler}\\
s_{t}= & q_{t}+k_{t}\\
x_{t+1}= & \phi_{x}x_{t}+\epsilon_{x,t+1}\\
h_{t+1}= & \phi_{h}h_{t}+\epsilon_{h,t+1}\\
\log\mathbb{E}_{t}\exp(a\epsilon_{x,t+1})= & \frac{1}{2}a^2 \sigma_x^2 h_t \label{eq:app_ezlin_CGFx}\\
\log\mathbb{E}_{t}\exp(a\epsilon_{h,t+1})= & \frac{1}{2}a^{2}\sigma_{h}^{2}\label{eq:app_ezlin_CGFh}
\end{align}

\subsection{Solving the model\label{subsec:app_ezmodel_Solving-the-model}}

We solve the model using the method of undetermined coefficients. The overview of the solver is as follows:
\begin{enumerate}
\item Conjecture linear laws of motion (LOMs) for $s_{t}$ and $v_{t}$.
\item Substitute out consumption, investment, and $q_{t}$ from the Euler
equation and value recursion and replace with $s_{t}$.
\item Apply the LOMs until the Euler equation is written in terms of state
variables today.
\item Collect terms and note that terms multiplying each state variable
must be zero. This logic provides a system of equations that provide
the LOM coefficients.
\end{enumerate}

\paragraph*{Solution for the stock price}

The LOM for log stock prices is
\begin{align*}
s_{t}= & Z_{s0}+Z_{sk}k_{t}+Z_{sx}x_{t}+Z_{sh}h_{t}.
\end{align*}
The elasticity with respect to capital $k_{t}$ is
\begin{align*}
Z_{sk}= & \frac{-\eta_{1}-\sqrt{\eta_{1}^{2}-4\eta_{0}\eta_{2}}}{2\eta_{2}},
\end{align*}
where

\begin{align*}
\eta_{2}= & \frac{\delta\xi\left(\xi\overline{\text{IC}}\bar{R}+\psi\right)}{\psi},\\
\eta_{1}= & \frac{\psi\left((\alpha-1)\alpha\delta\xi\overline{\text{YK}}-2\delta\xi+1\right)-\bar{R}\left(\alpha\delta\xi\overline{\text{YC}}+\delta\xi(2\xi-1)\overline{\text{IC}}+\psi\right)}{\psi},\\
\eta_{0}= & \frac{\bar{R}\left(\alpha\delta\xi\overline{\text{YC}}+\delta(\xi-1)\xi\overline{\text{IC}}+\psi\right)-\psi(\delta\xi-1)\left((\alpha-1)\alpha\overline{\text{YK}}-1\right)}{\psi}.
\end{align*}
The elasticity with respect to productivity $x_{t}$ is

\begin{align*}
Z_{sx}= & \frac{\bar{R}\overline{\text{YC}}\left(\phi_{x}-1\right)-\alpha\psi\overline{\text{YK}}\phi_{x}}{\psi\left((\alpha-1)\alpha\delta\xi\overline{\text{YK}}+\delta\xi\left(Z_{\text{sk}}-1\right)+\phi_{x}\right)-\bar{R}\left(\alpha\delta\xi\overline{\text{YC}}-\xi\overline{\text{IC}}\left(\delta\xi Z_{\text{sk}}+\delta(-\xi)+\delta+\phi_{x}-1\right)+\psi\right)}.
\end{align*}
The elasticity with respect to the high-order process $h_{t}$ is
\begin{align*}
Z_{sh}= & \frac{\psi\bar{R}\left((\gamma-\frac{1}{\psi})G_{1}\left(-(\gamma-1)Z_{\text{vx}},\sigma_{x}\right)-(\gamma-1)G_{1}\left(\frac{\alpha\overline{\text{YK}}+Z_{\text{sx}}}{\bar{R}}+\frac{\xi\overline{\text{IC}}Z_{\text{sx}}}{\psi}-\frac{\overline{\text{YC}}}{\psi}-\gamma Z_{\text{vx}}+\frac{1}{\psi}Z_{\text{vx}},\sigma_{x}\right)\right)}{(\gamma-1)\left(\psi\left((\alpha-1)\alpha\delta\xi\overline{\text{YK}}+\delta\xi\left(Z_{\text{sk}}-1\right)+\phi_{h}\right)-\bar{R}\left(\alpha\delta\xi\overline{\text{YC}}-\xi\overline{\text{IC}}\left(\delta\xi Z_{\text{sk}}+\delta(-\xi)+\delta+\phi_{h}-1\right)+\psi\right)\right)},
\end{align*}
where

\begin{align*}
Z_{vk}= & \frac{\left(\overline{\beta}-1\right)\left(\alpha\overline{\text{YC}}+\overline{\text{IC}}\left(\xi\left(-Z_{\text{sk}}\right)+\xi-1\right)\right)}{\delta\lambda\xi Z_{\text{sk}}-\delta\lambda\xi+\lambda-1},\\
Z_{\text{vx}}= & \frac{\xi Z_{\text{sx}}\left(-\overline{\beta}\overline{\text{IC}}+\overline{\text{IC}}-\delta\lambda Z_{\text{vk}}\right)+\left(\overline{\beta}-1\right)\overline{\text{YC}}}{\lambda\phi_{x}-1},
\end{align*}
and
\begin{align*}
1-\bar{\beta}= & (1-\beta)\left(\bar{C}/\bar{V}\right)^{(1-\psi^{-1})},\\
\lambda=\beta & \left(\bar{W}/\bar{V}\right)^{(1-\psi^{-1})}.
\end{align*}
The intercept is
\begin{align*}
Z_{s0}= & \omega_{s,1}/\omega_{s,2},
\end{align*}
where

\begin{align}
\omega_{s,1}= & 2(\gamma-1)\psi\bar{R}\left(\alpha\psi\overline{\text{YK}}\left((\alpha-1)\delta^{2}+1\right)+\sigma_{h}^{2}\left(-Z_{\text{sh}}\right)\left(\xi\overline{\text{IC}}Z_{\text{sh}}+(1-\gamma\psi)Z_{\text{vh}}\right)+\psi\left(\delta^{2}Z_{\text{sk}}-\delta^{2}+2\right)\right)\nonumber \\
 & -(\gamma-1)\psi^{2}\sigma_{h}^{2}Z_{\text{sh}}^{2}\nonumber \\
 & \bar{R}^{2}(\gamma-1)\left(-2\alpha\delta^{2}\psi\overline{\text{YC}}+2\psi\left(\psi\left(\log\left(\frac{\beta\bar{V}}{\bar{M}\bar{W}}\right)-1\right)+(\gamma\psi-1)\log\left(\frac{\bar{W}}{\bar{V}}\right)+\psi\log\left(\bar{M}\bar{R}\right)\right)-(\psi-1)(\gamma\psi-1)\sigma_{h}^{2}Z_{\text{vh}}^{2}\right)\nonumber \\
 & +2\bar{R}^{2}(\gamma-1)\overline{\text{IC}}\left(\delta^{2}\psi\left(\xi Z_{\text{sk}}-\xi+1\right)+\xi(\gamma\psi-1)\sigma_{h}^{2}Z_{\text{sh}}Z_{\text{vh}}\right)\nonumber \\
 & -\bar{R}^{2}(\gamma-1)\xi^{2}\overline{\text{IC}}^{2}\sigma_{h}^{2}Z_{\text{sh}}^{2}-2\bar{R}^{2}\psi(\gamma\psi-1)G_{0}\left(-(\gamma-1)Z_{\text{vx}},\sigma_{x}\right),\\
\omega_{s,2}= & -2(\gamma-1)\psi\bar{R}\left(\psi\left((\alpha-1)\alpha\delta\xi\overline{\text{YK}}+\delta\xi Z_{\text{sk}}-\delta\xi+1\right)-\bar{R}\left(\alpha\delta\xi\overline{\text{YC}}-\delta\xi\overline{\text{IC}}\left(\xi Z_{\text{sk}}-\xi+1\right)+\psi\right)\right).
\end{align}

\paragraph{Solution for capital}

The LOM for log capital is
\begin{align*}
k_{t+1}= & Z_{k0}+Z_{kk}k_{t}+Z_{kx}x_{t}+Z_{kh}h_{t}.
\end{align*}
The elasticity with respect to capital $k_{t}$ is
\begin{align*}
Z_{kk}= & \frac{-\eta_{k,1}-\sqrt{\eta_{k,1}^{2}-4\eta_{k,0}\eta_{k,2}}}{2\eta_{k,2}},
\end{align*}
where

\begin{align*}
\eta_{k,2}= & \frac{\xi\overline{\text{IC}}\bar{R}+\psi}{\delta\psi\xi},\\
\eta_{k,1}= & \frac{\psi\left((\alpha-1)\alpha\delta\xi\overline{\text{YK}}-1\right)-\bar{R}\left(\alpha\delta\xi\overline{\text{YC}}-(\delta-2)\xi\overline{\text{IC}}+\psi\right)}{\delta\psi\xi},\\
\eta_{k,0}= & \frac{\bar{R}\left(\alpha\delta\xi\overline{\text{YC}}+\overline{\text{IC}}(\xi-\delta\xi)+\psi\right)}{\delta\psi\xi}.
\end{align*}
The elasticity with respect to productivity $x_{t}$ is

\begin{align*}
Z_{kx}= & \frac{\delta\xi\left(\alpha\psi\overline{\text{YK}}\phi_{x}-\bar{R}\overline{\text{YC}}\left(\phi_{x}-1\right)\right)}{\bar{R}\left(\alpha\delta\xi\overline{\text{YC}}-\xi\overline{\text{IC}}\left(\delta+Z_{\text{kk}}+\phi_{x}-2\right)+\psi\right)-\psi\left((\alpha-1)\alpha\delta\xi\overline{\text{YK}}+Z_{\text{kk}}+\phi_{x}-1\right)}.
\end{align*}
The elasticity with respect to the high-order process $h_{t}$ is
\begin{align*}
Z_{kh}= & -\frac{\delta\xi\bar{R}\left((\gamma\psi-1)G_{1}\left(-(\gamma-1)Z_{\text{vx}},\sigma_{x}\right)-(\gamma-1)\psi G_{1}\left(\frac{\psi\left(\alpha\delta\xi\overline{\text{YK}}+Z_{\text{kx}}\right)+\xi\bar{R}\left(-\delta\overline{\text{YC}}+\overline{\text{IC}}Z_{\text{kx}}+\delta(1-\gamma\psi)Z_{\text{vx}}\right)}{\delta\psi\xi\bar{R}},\sigma_{x}\right)\right)}{(\gamma-1)\left(\bar{R}\left(\alpha\delta\xi\overline{\text{YC}}-\xi\overline{\text{IC}}\left(\delta+\phi_{h}+Z_{\text{kk}}-2\right)+\psi\right)-\psi\left((\alpha-1)\alpha\delta\xi\overline{\text{YK}}+\phi_{h}+Z_{\text{kk}}-1\right)\right)}.
\end{align*}
The intercept is
\begin{align*}
Z_{k0}= & \omega_{k,1}/\omega_{k,2},
\end{align*}
where

\begin{align*}
\omega_{k,1}= & \xi\bar{R}^{2}\delta^{2}((\gamma-1)(\psi-1)\xi(\gamma\psi-1)\sigma_{h}^{2}Z_{\text{vh}}^{2}-2\psi\\
 & \left((\gamma-1)\left(\psi\left(\xi\log\left(\frac{\beta\bar{V}}{\bar{M}\bar{W}}\right)+\delta-\xi\right)+\xi(\gamma\psi-1)\log\left(\frac{\bar{W}}{\bar{V}}\right)+\psi\xi\log\left(\bar{M}\bar{R}\right)\right)+(\xi-\gamma\psi\xi)G_{0}\left(-(\gamma-1)Z_{\text{vx}},\sigma_{x}\right)\right)\\
 & -\xi\bar{R}^{2}2\delta(\gamma-1)\xi\overline{\text{IC}}(\gamma\psi-1)\sigma_{h}^{2}Z_{\text{kh}}Z_{\text{vh}}+\xi\bar{R}^{2}(\gamma-1)\xi\overline{\text{IC}}^{2}\sigma_{h}^{2}Z_{\text{kh}}^{2}\\
 & +2(\gamma-1)\psi\xi\bar{R}\left(-\alpha\delta^{2}\psi\xi\overline{\text{YK}}+\overline{\text{IC}}\sigma_{h}^{2}Z_{\text{kh}}^{2}+\delta\left(\delta\psi(\delta-2\xi)+(1-\gamma\psi)\sigma_{h}^{2}Z_{\text{kh}}Z_{\text{vh}}\right)\right)+(\gamma-1)\psi^{2}\sigma_{h}^{2}Z_{\text{kh}}^{2},\\
\omega_{k,2}= & -2\delta(\gamma-1)\psi\xi\bar{R}\left(\bar{R}\left(\alpha\delta\xi\overline{\text{YC}}-\xi\overline{\text{IC}}\left(\delta+Z_{\text{kk}}-1\right)+\psi\right)-\psi\left((\alpha-1)\alpha\delta\xi\overline{\text{YK}}+Z_{\text{kk}}\right)\right).
\end{align*}

\section{Proof of Theorem \ref{prop:main_theorem} \label{subsec:appendix_abstract}}

The essence of the theorem is that the method of undetermined coefficients allows us to solve for the elasticities with respect to first moment and higher moment variables in separate layers, while keeping track of the parameters involved in determining these elasticities. The irrelevance theorem comes down to proving that the elasticities with respect to first moment variables\textemdash and those with respect to lagged endogenous variables\textemdash do not depend on risk aversion and high moment parameters. The first layer contains the endogenous dynamics: the elasticities of endogenous variables with respect to endogenous (lagged) states. We show these endogenous dynamics are determined irrespective of risk aversion and the parameters that govern exogenous risk dynamics, extending the results of \citet{schmitt2004solving} to a setting with recursive preferences and time-varying risk. The second layer of the solution consists of the elasticities of endogenous variables with respect to first moment states. In this layer as well the elasticities are determined irrespective of both risk aversion and time-varying risk parameters. These first two layers provide us with our main irrelevance theorem. The final layer of the model concerns the elasticities of endogenous variables with respect to higher moment states, which are states that directly affect only the higher moments of the model. In this layer, the elasticities depend on all aspects of the model, including risk aversion and time-varying risk.
  
The proof proceeds in three steps. Section \ref{sec:mainproof_lemmas} provides some preliminary lemmas. Section \ref{sec:mainproof_proposition} uses them to characterize the main elasticity matrices. With that characterization in hand, Section \ref{sec:mainproof_theorem} proves the theorem.

\subsection{Preliminary Lemmas}\label{sec:mainproof_lemmas}

We begin by showing that the $\mathbb{H}_{t}$ operator preserves linearity due to the fact that the CGF of the shock terms (Equations (\ref{eq:main_aff_epx}) and (\ref{eq:main_aff_eph})) are approximately linear.
\begin{lem}
\label{prop:Hzx}Let $\boldsymbol{x_{t+1}}$ and $\boldsymbol{z_{t+1}}$
follow the processes in Equations (\ref{eq:main_x_process}) and (\ref{eq:main_aff_solution-1}),
respectively. For any $n>0,$ any $n\times n_{x}$ matrix $\nu_{x}$,
and any $n\times n_{z}$ matrix $\nu_{z}$, the entropy of the linear
combination $\nu_{z}\mathbf{z}_{t+1}+\nu_{x}\mathbf{x}_{t+1}$ is given
by
\begin{eqnarray}
	\label{eq:app_main_Hzx}
\mathbb{H}_{t}\left[\nu_{z}\mathbf{z}_{t+1}+\nu_{x}\mathbf{x}_{t+1}\right] & = & \hat{g}_{xh}(\nu_{z}Z_{x}+\nu_{x},\nu_{z}Z_{h};\Theta) + 
  \hat{G}_{xh}(\nu_{z}Z_{x}+\nu_{x},\nu_{z}Z_{h};\Theta)\mathbf{h}_{t},
\end{eqnarray}
where
\begin{align}
\hat{g}_{xh}(\nu_{z}Z_{x}+\nu_{x},\nu_{z}Z_{h};\Theta)\equiv & \tilde{g}_{x}(\nu_{z}Z_{x}+\nu_{x};\Theta_{h})+\tilde{g}_{h}(\nu_{z}Z_{h};\Theta),\label{eq:app_main_gxh0}\\
\hat{G}_{xh}(\nu_{z}Z_{x}+\nu_{x},\nu_{z}Z_{h};\theta_{x},\theta_{h})\equiv & \tilde{G}_{x}(\nu_{z}Z_{x}+\nu_{x};\theta_{x},\theta_{h})+\tilde{G}_{h}(\nu_{z}Z_{h};\theta_{h})\label{eq:app_main_Gxh}
\end{align}
and $\tilde{g}_{x}()$, $\tilde{g}_{h}()$, $\tilde{G}_{x}()$, and $\tilde{G}_{h}()$ as defined in Equations (\ref{eq:main_aff_epx}) and (\ref{eq:main_aff_eph}). \end{lem}
\begin{proof}
Apply the definition of $\mathbb{H}_{t}$
in Equation (\ref{eqn:curvature}), replacing the laws of motion for
$\boldsymbol{z_{t+1}}$, and $\boldsymbol{x_{t+1}}$, and using the
linear approximations of the CGFs of $\epsilon_{x,t+1}$ and $\epsilon_{h,t+1}$
in Assumption \ref{def:cond_affine_CGF} and the fact that the two
vectors of random variables are independent.
\end{proof}

Plugging in the conjectured solution (Equation \eqref{eq:main_aff_solution-1}) into Proposition \ref{prop:all_plugged_in-1}, and applying Lemma \ref{prop:Hzx}, leads to a more detailed partitioning of the equilibrium conditions.
\begin{lem}
\label{lem:concise_equation}
If Equation \eqref{eq:main_aff_solution-1} holds, then Equations (\ref{eq:main_aff_tech-1}) and (\ref{eq:main_Euler-1})
can be written as
\begin{eqnarray}
b^{1}(\theta_{0})+B_{z}^{1}(\theta_{0})\mathbf{z}_{t}+B_{zl}^{1}(\theta_{0})\mathbf{z}_{t-1}+B_{x}^{1}(\theta_{0})\mathbf{x}_{t} & = & \mathbb{E}_{t}\left[D_{z}^{1}(\theta_{0})\mathbf{z}_{t+1}+D_{x}^{1}(\theta_{0})\mathbf{x}_{t+1}\right]\nonumber \\
  & + 
  & g^{1}(Z_{x},Z_{h};\theta_{0},\theta_{h})
   +G_{h}^{1}(Z_{x},Z_{h};\theta_{0},\theta_{h})\mathbf{h}_{t},\ \ \ \ \ \ \ \ \label{eq:main_aff_tech-1-1}
\end{eqnarray}
and
\begin{align}
\boldsymbol{0}_{n_{m}}= & \mathbb{E}_{t}[D_{z}^{2^{*}}(\theta_{0})\mathbf{z}_{t+1}]+g^{2}(Z_{x},Z_{h};\Theta)+G_{h}^{2}(Z_{x},Z_{h};\Theta)\mathbf{h}_{t},\label{eq:main_Euler-1-1}
\end{align}
respectively, where $g^{1}(Z_{x},Z_{h};\theta_{0},\theta_{h})$ and $G_{h}^{1}(Z_{x},Z_{h};\theta_{0},\theta_{h})$ are the constant and slope terms obtained by applying Assumption \ref{def:cond_affine_CGF} to the entropy term in Equation \eqref{eq:main_aff_tech-1}, and $g^{2}(Z_{x},Z_{h};\Theta)$ and $G_{h}^{2}(Z_{x},Z_{h};\Theta)$ are the constant and slope terms obtained by applying Assumption \ref{def:cond_affine_CGF} to the entropy term in Equation \eqref{eq:main_Euler-1}.

Therefore, the equilibrium equations can be expressed concisely as
\begin{eqnarray}
b(\theta_{0})+B_{z}(\theta_{0})\mathbf{z}_{t}+B_{zl}(\theta_{0})\mathbf{z}_{t-1}+B_{x}(\theta_{0})\mathbf{x}_{t} & = & \mathbb{E}_{t}\left[D_{z}(\theta_{0})\mathbf{z}_{t+1}+D_{x}(\theta_{0})\mathbf{x}_{t+1}\right]\nonumber \\
  & + & g(Z_{x},Z_{h};\Theta)+G_{h}(Z_{x},Z_{h};\Theta)\mathbf{h}_{t},\ \ \ \label{eq:app_main_equilibrium-1}
\end{eqnarray}
where
\begin{align*}
  b(\theta_{0}) & \equiv\left(\begin{array}{c}
    b^{1}(\theta_{0})\\ \nonumber
    \boldsymbol{0}
  \end{array}\right),\quad B_{z}(\theta_{0})\equiv\left(\begin{array}{c}
B_{z}^{1}(\theta_{0}) , \\ \boldsymbol{0} \end{array}\right),  \quad B_{zl}(\theta_{0})  \equiv\left(\begin{array}{c} B_{zl}^{1}(\theta_{0})\\ \boldsymbol{0} \end{array}\right),\\ B_{x}(\theta_{0}) & \equiv\left(\begin{array}{c} B_{x}^{1}(\theta_{0})\\ \boldsymbol{0} \end{array}\right), \quad	D_{z}(\theta_{0})  \equiv\left(\begin{array}{c} D_{z}^{1}(\theta_{0})\\ D_{z}^{2^{*}}(\theta_{0}) \end{array}\right),\quad D_{x}(\theta_{0})\equiv\left(\begin{array}{c} D_{x}^{1}(\theta_{0})\\ \boldsymbol{0} \end{array}\right),\\ g(Z_{x},Z_{h};\Theta)\equiv & \left(\begin{array}{c} g^{1}(Z_{x},Z_{h};\theta_{0},\theta_{h})\\ g^{2}(Z_{x},Z_{h};\Theta) \end{array}\right)\text{, \ \ and \ \ }G_{h}(Z_{x},Z_{h};\Theta)\equiv\left(\begin{array}{c} G_{h}^{1}(Z_{x},Z_{h};\theta_{0},\theta_{h})\\ G_{h}^{2}(Z_{x},Z_{h};\Theta) \end{array}\right). \nonumber 
\end{align*}
\end{lem} 

\begin{proof}
The proof relies on showing that, given the conjectured solution in Equation (\ref{eq:main_aff_solution-1}), the $\mathbb{H}_{t}$ terms in Equations (\ref{eq:main_aff_tech-1}) and (\ref{eq:main_Euler-1}) can be expressed as linear combinations of the high order variables $\mathbf{h}_{t}$ by using Assumption \ref{def:cond_affine_CGF} and some algebraic manipulations. Once this is shown, the concise system of equilibrium conditions in Equation (\ref{eq:app_main_equilibrium-1}) follows.

We apply Lemma \ref{prop:Hzx} to the $\mathbb{H}_{t}$ terms in Equations (\ref{eq:main_aff_tech-1}) and (\ref{eq:main_Euler-1}) to express them, respectively, as
\begin{eqnarray}
\mathbb{H}_{t}\left[D_{z}^{1}(\theta_{0})\mathbf{z}_{t+1}+D_{x}^{1}(\theta_{0})\mathbf{x}_{t+1}\right] & = & g^{1}(Z_{x},Z_{h};\Theta)+G_{h}^{1}(Z_{x},Z_{h};\Theta)\mathbf{h}_{t}\nonumber \\
 & = & \hat{g}_{xh}(D_{z}^{1}(\theta_{0})Z_{x}+D_{x}^{1}(\theta_{0}),D_{z}^{1}(\theta_{0})Z_{h};\Theta_{x})\nonumber \\
 & + & \hat{G}_{xh}(D_{z}^{1}(\theta_{0})Z_{x}+D_{x}^{1}(\theta_{0}),D_{z}^{1}(\theta_{0})Z_{h};\theta_{x},\theta_{h})\mathbf{h}_{t},\ \ \label{eq:app_main_gm0}
\end{eqnarray}
and

\begin{eqnarray}
- & \left(\frac{\frac{1}{\psi}-\gamma}{1-\gamma}\right)\mathbb{H}_{t}\left[(1-\gamma)A_{v}\text{\textbf{(\ensuremath{\theta_{0}})}}\mathbf{z}_{t+1}\right]\boldsymbol{1}_{n_{m}}+\mathbb{H}_{t}\left[D_{z}^{2}(\theta_{0})\mathbf{z}_{t+1}\right]\nonumber \\
= & g^{2}(Z_{x},Z_{h};\Theta)+G_{h}^{2}(Z_{x},Z_{h};\Theta)\mathbf{h}_{t}\nonumber \\
= & -\left(\frac{\frac{1}{\psi}-\gamma}{1-\gamma}\right)\hat{g}_{xh}((1-\gamma)A_{v}\text{\textbf{(\ensuremath{\theta_{0}})}}Z_{x},(1-\gamma)A_{v}Z_{h};\Theta)\boldsymbol{1}_{n_{m}}\nonumber \\
 & +\hat{g}_{xh}(D_{z}^{2}(\theta_{0})Z_{x},D_{z}^{2}(\theta_{0})Z_{h};\Theta)\nonumber \\
 & -\left(\frac{\frac{1}{\psi}-\gamma}{1-\gamma}\right)\hat{G}_{xh}((1-\gamma)A_{v}\text{\textbf{(\ensuremath{\theta_{0}})}}Z_{x},(1-\gamma)A_{v}Z_{h};\Theta)\boldsymbol{1}_{n_{m}}\mathbf{h}_{t}\nonumber \\
 & +\hat{G}_{xh}(D_{z}^{2}(\theta_{0})Z_{x},D_{z}^{2}(\theta_{0})Z_{h};\Theta)\mathbf{h}_{t}.\label{eq:app_main_Gmh}
\end{eqnarray}

Replacing Equation (\ref{eq:app_main_gm0}) in Equation (\ref{eq:main_aff_tech-1}) and Equation (\ref{eq:app_main_Gmh}) in Equation (\ref{eq:main_Euler-1}), we obtain Equations (\ref{eq:main_aff_tech-1-1}) and (\ref{eq:main_Euler-1-1}), respectively. Grouping corresponding terms in these equations results in Equation (\ref{eq:app_main_equilibrium-1}).
\end{proof}

\subsection{Characterization of the Elasticities}\label{sec:mainproof_proposition}

\begin{prop}
\label{prop:layered_solution}Lemma \ref{lem:concise_equation} implies
that the coefficients and elasticities in the conjectured endogenous
process in Equation (\ref{eq:main_aff_solution-1}) that solve the
system of equations (\ref{eq:app_main_equilibrium-1}) are obtained
as follows:

1. The elasticities $Z_{z}$ satisfy the quadratic matrix equation
\begin{equation}
D_{z}(\theta_{0})Z_{z}^{2}-B_{zl}(\theta_{0})-B_{z}(\theta_{0})Z_{z}=\boldsymbol{0}_{n_{z}\times n_{z}},\label{eq:main_4cond_Zz}
\end{equation}
\qquad{}\ where $\boldsymbol{0}_{n_{z}\times n_{z}}$ is an $n_{z}\times n_{z}$ matrix of zeros.

2. The elasticities $Z_{x}$ satisfy the linear system
\begin{equation}
(B_{z}(\theta_{0})-D_{z}(\theta_{0})Z_{z})Z_{x}-D_{z}(\theta_{0})Z_{x}\Phi_{x}=D_{x}(\theta_{0})\Phi_{x}-B_{x}(\theta_{0}).\label{eq:main_4cond_Zx}
\end{equation}

3. The elasticities $Z_{h}$ satisfy the system
\begin{equation}
(B_{z}(\theta_{0})-D_{z}(\theta_{0})Z_{z})Z_{h}=D_{z}(\theta_{0})Z_{h}\Phi_{h}+G_{h}(Z_{x},Z_{h};\Theta).\label{eq:main_4cond_Zh}
\end{equation}

4. The coefficients $z_{0}$ satisfy the linear system
\begin{equation}
(B_{z}(\theta_{0})-D_{z}(\theta_{0})Z_{z}-D_{z}(\theta_{0}))z+b_{0}=(D_{z}(\theta_{0})Z_{x}+D_{x}(\theta_{0}))\mu_{x}+g_{0}(Z_{x},Z_{h};\Theta).\label{eq:main_4cond_z0}
\end{equation}
\end{prop} 

\begin{proof}
Proving the proposition comes down to plugging in the conjectured solution (\ref{eq:main_aff_solution-1}) and the exogenous process for $\mathbf{x}_{t}$ in Equation (\ref{eq:main_x_process}) into Equation (\ref{eq:app_main_equilibrium-1}) and applying the method of undetermined coefficients in separate layers for the terms in front of $\text{\textbf{z}}_{t-1}$, $\text{\textbf{x}}_{t}$, $\text{\textbf{h}}_{t},$ and the constants, respectively.


We use Equations (\ref{eq:main_x_process}), (\ref{eq:main_h_process}), and (\ref{eq:main_aff_solution-1}) to express Equation (\ref{eq:app_main_equilibrium-1}) in terms of $\text{\textbf{x}}_{t}$, $\text{\textbf{h}}_{t}$, and $\text{\textbf{z}}_{t-1}$. The left-hand side of the equation becomes
\begin{eqnarray}
b(\theta_{0})+B_{z}(\theta_{0})\mathbf{z}_{t}+B_{zl}(\theta_{0})\mathbf{z}_{t-1}+B_{x}(\theta_{0})\mathbf{x}_{t} & = & b+B_{z}(\theta_{0})z+\left(B_{z}(\theta_{0})Z_{z}+B_{zl}(\theta_{0})\right)\mathbf{z}_{t-1}\nonumber \\
 &  & +\left(B_{z}(\theta_{0})Z_{x}+B_{x}(\theta_{0})\right)\mathbf{x}_{t}+B_{z}(\theta_{0})Z_{h}\mathbf{h}_{t}.\quad\quad\quad\quad\label{eq:app_main_eqm_lhs}
\end{eqnarray}
Similarly, the first term on the right-hand side of Equation (\ref{eq:app_main_equilibrium-1}) becomes
\begin{eqnarray}
\mathbb{E}_{t}\left[D_{z}(\theta_{0})\mathbf{z}_{t+1}+D_{x}(\theta_{0})\mathbf{x}_{t+1}\right] & = & D_{z}(\theta_{0})\left(\boldsymbol{1}+Z_{z}\right)z+(D_{z}(\theta_{0})Z_{x}+D_{x}(\theta_{0}))\mu_{x}+D_{z}(\theta_{0})Z_{z}^{2}\mathbf{z}_{t-1}\nonumber \\
 & + & \left[D_{z}(\theta_{0})Z_{z}Z_{x}+(D_{z}(\theta_{0})Z_{x}+D_{x}(\theta_{0}))\Phi_{x}\right]\mathbf{x}_{t}\nonumber \\
 &  & +\left[D_{z}(\theta_{0})Z_{z}Z_{h}+D_{z}(\theta_{0})Z_{h}\Phi_{h}\right]\mathbf{h}_{t}.\label{eq:app_main_eqm_rhs1}
\end{eqnarray}
Replacing Equations (\ref{eq:app_main_eqm_lhs}) and (\ref{eq:app_main_eqm_rhs1}) into (\ref{eq:app_main_equilibrium-1}) results in
\begin{align}
\boldsymbol{0}_{n_{z}\times1}= & \left[D_{z}(\theta_{0})Z_{z}^{2}-B_{zl}(\theta_{0})-B_{z}(\theta_{0})Z_{z}\right]\mathbf{z}_{t-1}\nonumber \\
 & +\left[\left[-\hat{B}_{z}(\theta_{0})+D_{z}(\theta_{0})Z_{z}\right]Z_{x}+D_{z}(\theta_{0})\Phi_{x}Z_{x}-B_{x}(\theta_{0})+D_{x}(\theta_{0})\Phi_{x})\right]\text{\textbf{x}}_{t}\nonumber \\
 & +\left[-B_{z}(\theta_{0})Z_{h}+D_{z}(\theta_{0})Z_{z}Z_{h}+D_{z}(\theta_{0})Z_{h}\Phi_{h}+G_{h}(Z_{x},Z_{h};\Theta)\right]\mathbf{h}_{t}\nonumber \\
 & -b(\theta_{0})-B_{z}(\theta_{0})z+D_{z}(\theta_{0})\left(\boldsymbol{1}+Z_{z}\right)z+(D_{z}(\theta_{0})Z_{x}+D_{x}(\theta_{0}))\mu_{x}\nonumber \\
 & +g(Z_{x},Z_{h};\Theta).\label{eq:all_plugged_in}
\end{align}

Using the method of undetermined coefficients, Equation (\ref{eq:all_plugged_in}) tells us that all terms in front of the state variables and the constants should equal zero to satisfy the equilibrium conditions. In turn, the method allows us to find the coefficients $z$, $Z_{z}$, $Z_{x}$ and $Z_{h}$ in the solution (\ref{eq:main_aff_solution-1}) in separate layers. 

First, the coefficients loading on the lagged endogenous variables $\mathbf{z}_{t-1}$ satisfy:
\begin{equation}
D_{z}(\theta_{0})Z_{z}^{2}-B_{zl}(\theta_{0})-B_{z}(\theta_{0})Z_{z}=\boldsymbol{0}_{n_{z}\times n_{z}}.\label{eq:app_main_4cond_Zz}
\end{equation}
Solving this quadratic matrix equation provides the $n_{z}\times n_{z}$ coefficients $Z_{z}$. If the number of predetermined endogenous variables is the same as the total number of endogenous variables $n_{z}$, the solution can be found using the methods described in \citet{mccallum1983non} or \citet{uhlig1999toolkit}. If the number of predetermined variables is less than $n_{z}$, the solution of the quadratic matrix equation can be found numerically, under the restriction that, if all the elements in a column of $B_{zl}$ are zero, all the elements in the corresponding column in $Z_{z}$ are zero too. Using the minimum state variable criterion, $Z_{z}=\boldsymbol{0}$ if $B_{zl}\equiv\boldsymbol{0}$.

Second, the coefficients on $\mathbf{x}_{t}$ in Equation (\ref{eq:all_plugged_in})
satisfy the linear system
\begin{equation}
(B_{z}(\theta_{0})-D_{z}(\theta_{0})Z_{z})Z_{x}-D_{z}(\theta_{0})Z_{x}\Phi_{x}=D_{x}(\theta_{0})\Phi_{x}-B_{x},\label{eq:app_main_4cond_Zx}
\end{equation}
which is a Sylvester equation that can be written as
\begin{eqnarray*}
vec(Z_{x})=\left(\mathbb{I}_{n_{x}}\otimes(B_{z}(\theta_{0})-D_{z}(\theta_{0})Z_{z})-\Phi_{x}^{\top}\otimes D_{z}\right)^{-1}\text{vec}\left(D_{x}\Phi_{x}-B_{x}(\theta_{0})\right),
\end{eqnarray*}
where the operator vec($W$) denotes the vectorization of matrix $W$, $\mathbb{I}_{n}$ is the $n\times n$ identity matrix, $\otimes$ denotes the Kronecker product. The Sylvester equation leads to the $n_{z}\times n_{s}$ coefficients $Z_{x}$.

Third, the coefficients on $\mathbf{h}_{t}$ in equation (\ref{eq:main_aff_solution-1}) satisfy \[ (B_{z}-D_{z}Z_{z})Z_{h}=D_{z}Z_{h}\Phi_{h}+G_{h}(Z_{x},Z_{h};\Theta), \] which is, in general, a nonlinear equation on $Z_{h}$. Note that for the particular case of a linear $G_{h}(Z_{x},Z_{h};\Theta)$ on $Z_{h}$, the equation becomes a Sylvester equation that can be written as
\begin{eqnarray*}
\text{vec}(Z_{h})=\left(\mathbb{I}_{n_{h}}\otimes(B_{z}-D_{z}Z_{z})-\Phi_{h}^{\top}\otimes D_{z}\right)^{-1}\text{vec}\left(G_{h}\right).
\end{eqnarray*}
Using the minimum state variable criterion, $Z_{h}=\boldsymbol{0}$ if $\Phi_{h}=B_{h}=G_{h}\equiv\boldsymbol{0}$.

Finally, the constant coefficients in equation (\ref{eq:main_aff_solution-1}) imply the system of $n_{z}$ linear equations \[ (B_{z}-D_{z}Z_{z}-D_{z})z+b=(D_{z}Z_{x}+D_{x})\mu_{x}+g(Z_{x},Z_{h};\Theta). \]

\end{proof}

\subsection{Proof of Theorem \ref{prop:main_theorem}}\label{sec:mainproof_theorem}

\begin{proof}
  Given Proposition \ref{prop:layered_solution}, the proof easily follows from noticing that the terms in front of
  $\text{\textbf{z}}_{t-1}$ in Equation (\ref{eq:main_4cond_Zz}) that
  determine $Z_{z}$ are functions of $\theta_{0}$ only. Similarly, given $Z_{z}$, Equation (\ref{eq:main_4cond_Zx})
  tells us that the terms in front of $\text{\textbf{x}}_{t}$ determine
  $Z_{x}$ entirely with functions of $\theta_{0}$. 
\end{proof}

\section{Barro-King with Recursive Preferences \label{subsec:app_barro_king}}

We first show that the labor-leisure condition under Epstein-Zin preferences is the same as in time-separable preferences. Consider Epstein-Zin preferences
\begin{align*}
V_{t}= & \left\{ (1-\beta)u(C_{t},1-L_{t})^{1-1/\psi}+\beta P_{t}^{1-1/\psi}\right\} ,
\end{align*}
where $L_{t}$ is labor. The first order condition for consumption implies
\begin{align}
\lambda= & \frac{1}{(1-\frac{1}{\psi})}\left(V_{t}^{1-\frac{1}{\psi}}\right)^{\frac{1}{1-\frac{1}{\psi}}-1}(1-\beta)(1-\frac{1}{\psi})u(C_{t},1-L_{t})^{-\frac{1}{\psi}}\frac{\partial u}{\partial C_{t}},\label{eq:app_bk_EZ1}
\end{align}
where $\lambda$ be the Lagrange multiplier on the household's budget constraint. The first order condition for labor is
\begin{align}
\lambda W_{t}= & \frac{1}{(1-\frac{1}{\psi})}\left(V_{t}^{1-\frac{1}{\psi}}\right)^{\frac{1}{1-\frac{1}{\psi}}-1}(1-\beta)(1-\frac{1}{\psi})u(C_{t},1-L_{t})^{-\frac{1}{\psi}}\frac{\partial u}{\partial L_{t}},\label{eq:app_bk_EZ2}
\end{align}
where $W_{t}$ is the wage. Substituting (\ref{eq:app_bk_EZ1}) into (\ref{eq:app_bk_EZ2}), we have the usual labor-leisure condition
\begin{align*}
W_{t}= & \frac{\frac{\partial u}{\partial L_{t}}\left[\frac{V_{t}}{u(C_{t},1-L_{t})}\right]^{\frac{1}{\psi}}}{\frac{\partial u}{\partial C_{t}}\left[\frac{V_{t}}{u(C_{t},1-L_{t})}\right]^{\frac{1}{\psi}}}=\frac{u_{2}(C_{t},1-L_{t})}{u_{1}(C_{t},1-L_{t})}.
\end{align*}
In what follows, we drop the time subscripts since we are working with intratemporal concepts.

We now provide a formal proof of Barro-King impossibility. To provide this proof, one needs a few additional assumptions. Specifically, assume
\begin{itemize}
\item $s$ = a vector of state variables,
\item $x$ = productivity, one of the states in $s$,
\item $k$ = capital, also one of the states in $s$,
\item $F(x,k,L)$ = production.
\end{itemize}
Also assume that in equilibrium a resource constraint holds in every state
\begin{align}
C(s)+I(s) & =F(x,k,L(s)),\label{eq:app_bk_assump1}\\
\nonumber
\end{align}
and labor is paid its marginal product in every state
\begin{align*}
MRS(C(s),L(s))= & \frac{u_{2}(C(s),1-L(s))}{u_{1}(C(s),1-L(s))}=F_{3}(x,k,L(s)).
\end{align*}
In addition, the following restrictions on the functional forms hold:
\begin{align}
F_{33}<0 & \quad\text{(diminishing marginal product of labor)},\\
u_{12}u_{1}>u_{2}u_{11} & \quad\text{(consumption is a normal good)},\\
u_{12}u_{2}>u_{1}u_{22} & \quad\text{(leisure is a normal good)}.\label{eq:app_bk_assump_end}
\end{align}
These last two inequalities are equivalent to requiring the marginal rate of substitution $u_{2}/u_{1}$ to be decreasing in consumption and increasing in labor, so they impose the economically intuitive condition that both goods are normal.

Now consider a state $h$ that is neither $x$ nor $k$, and suppose without loss of generality that
\begin{align}
\frac{\partial L(s)}{\partial h}\ge0.\label{eq:app_bk_dldz}
\end{align}
In order for $h$ to generate comovement, we need also that
\begin{align}
\frac{\partial C(s)}{\partial h}\ge0,\quad\text{and}\quad\frac{\partial I(s)}{\partial h}\ge0,\label{eq:app_bk_dcdz}
\end{align}
and that at least one derivative is positive.

One can show that Equations (\ref{eq:app_bk_dldz}) and (\ref{eq:app_bk_dcdz}) cannot both be true. Moreover, one of those derivatives has to be negative.

To prove this, first consider the case that
\begin{align}
\frac{\partial L(s)}{\partial h}=0.
\end{align}
This case looks a bit trivial, but this is true for many of production based asset pricing models since most models don't consider leisure preferences. In this case, $\frac{\partial}{\partial h}F(x,k,L(s))=0$ (changing $h$ doesn't change production, since it doesn't affect labor), so the differentiating the resource constraint yields
\begin{align}
\frac{\partial C(s)}{\partial h}+\frac{\partial I(s)}{\partial h}=0,
\end{align}
and so either these derivatives have opposite signs or they're both zero.

Now consider the case with leisure preferences. Without loss of generality, we assume
\begin{align}
\frac{\partial L(s)}{\partial h}>0.
\end{align}
Now we differentiate the labor leisure condition with respect to $h$. This is a little complicated, so let's take it in steps. First, let's look at the big picture:
\begin{align}
MRS_{1}\frac{\partial C(s)}{\partial h}+MRS_{2}\frac{\partial L(s)}{\partial h}=F_{33}(x,k,L(s))\frac{\partial L(s)}{\partial h}<0.
\end{align}
Note that if $MRS_{1}>0$ and $MRS_{2}>0$, then $\frac{\partial C(s)}{\partial h}<0$, and we have proved impossibility.

First let's show $MRS_{1}>0$:
\begin{align}
MRS_{1} & =\frac{\partial}{\partial C(s)}\left[\frac{u_{2}(C(s),1-L(s))}{u_{1}(C(s),1-L(s))}\right]\\
 & =\frac{u_{12}u_{1}-u_{2}u_{11}}{(u_{1})^{2}}>0,
\end{align}
by the normal-good condition $u_{12}u_{1}>u_{2}u_{11}$. Similarly, one can show $MRS_{2}>0$:
\begin{align}
MRS_{2} & =\frac{\partial}{\partial L(s)}\left[\frac{u_{2}(C(s),1-L(s))}{u_{1}(C(s),1-L(s))}\right]\\
 & =\frac{u_{12}u_{2}-u_{1}u_{22}}{(u_{1})^{2}}>0,
\end{align}
using $u_{12}u_{2}>u_{1}u_{22}$. 
Thus, under the neoclassical structure (\ref{eq:app_bk_assump1})-(\ref{eq:app_bk_assump_end}) only productivity shocks can drive positive comovement in consumption, investment, and labor. This result restates the Barro and King (1984) comovement restriction using the weaker and economically interpretable normal-good conditions.

\section{The Need for 71 Gridpoints in Productivity \label{subsec:Nx_demo}}

\begin{figure}[tph]
\begin{centering}
\includegraphics[scale=0.5]{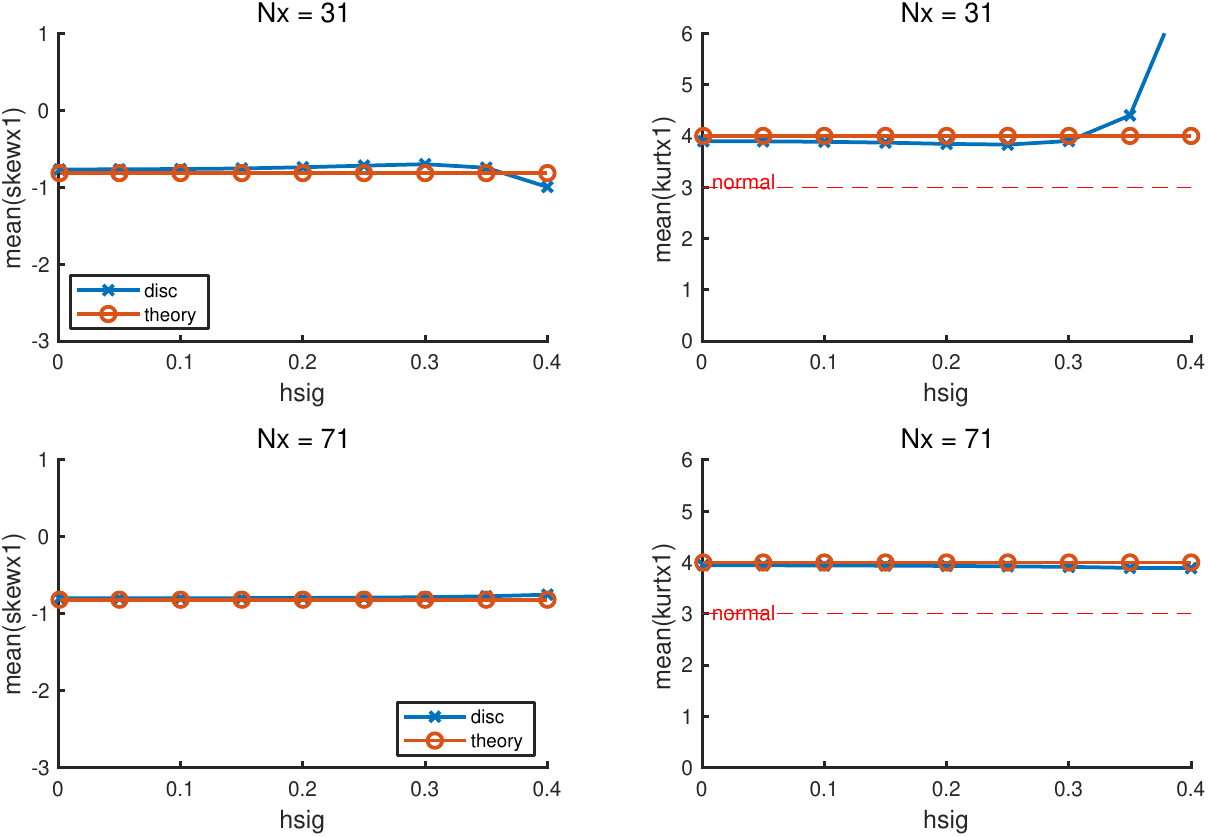}
\par\end{centering}
\textbf{\label{fig:Nx_demo}{\color{ChadBlue}The Accuracy of the Discretization of a Process with Time-Varying Fat Tails.}}
The top panels use 31 gridpoints. The discretization matches the theoretical
moment for most values of the conditional vol of vol, but for high
values the average kurtosis deviates significantly. The bottom panels
use 71 gridpoints and lead to a good approximation for all values
of the conditional vol of vol.
\end{figure}

\newpage{}

\newpage{}

\bibliographystyle{elsarticle-harv}
\bibliography{irrbib}

@incollection{Bernanke1999,
	author = {Bernanke, Ben S. and Gertler, Mark and Gilchrist, Simon},
	title = {The financial accelerator in a quantitative business cycle framework},
	editor = {Taylor, J. B. and Woodford, M.},
	booktitle = {Handbook of Macroeconomics},
	series = {Handbook of Macroeconomics},
	edition = {1},
	volume = {1},
	chapter = {21},
	pages = {1341-1393},
	year = {1999},
	publisher = {Elsevier}
}

@techreport{GourioNgo2020ZLB,
	author = {Gourio, Fran\c{c}ois and Ngo, Phuong},
	year = {2020},
	title = {Risk Premia at the ZLB: A Macroeconomic Interpretation},
	number = {WP 2020-01},
	series = {Working Paper Series},
	institution = {Federal Reserve Bank of Chicago},
	doi = {10.21033/wp-2020-01},
	month = {January}
}

@techreport{Zeke2024RiskTaking,
	author = {Zeke, David and Joel David},
	title = {Risk-Taking, Capital Allocation and Monetary Policy},
	institution = {Federal Reserve Bank of Chicago},
	type = {Working Paper},
	number = {WP-2021-01},
	year = {2024},
	month = {January},
	note = {Available at SSRN: https://ssrn.com/abstract=3856385}
}

@techreport{Basu2024Risky,
	author = {Basu, Susanto and Candian, Giacomo and Chahrour, Ryan and Valchev, Rosen},
	title = {Risky Business Cycles},
	institution = {National Bureau of Economic Research},
	series = {Working Paper Series},
	number = {28693},
	year = {2024},
	month = {Sep},
	url = {http://www.nber.org/papers/w28693}
}

@article{DiTellaHall2022,
	author={Di Tella, Sebastian and Hall, Robert E.},
	title={Risk premium shocks can create inefficient recessions},
	journal={The Review of Economic Studies},
	volume={89},
	number={3},
	pages={1335-1369},
	year={2022},
	publisher={Review of Economic Studies Ltd}
}

@article{cochrane2021rethinking,
  title={Rethinking production under uncertainty},
  author={Cochrane, John H},
  journal={The Review of Asset Pricing Studies},
  volume={11},
  number={1},
  pages={1--59},
  year={2021},
  publisher={Oxford University Press}
}

@article{weil1990nonexpected,
  title={Nonexpected utility in macroeconomics},
  author={Weil, Philippe},
  journal={The Quarterly Journal of Economics},
  volume={105},
  number={1},
  pages={29--42},
  year={1990},
  publisher={MIT Press}
}

@article{hall1988intertemporal,
  title={Intertemporal substitution in consumption},
  author={Hall, Robert E},
  journal={Journal of political economy},
  volume={96},
  number={2},
  pages={339--357},
  year={1988},
  publisher={The University of Chicago Press}
}

@article{pohl2018higher,
  title={Higher Order Effects in Asset Pricing Models with Long-Run Risks},
  author={Pohl, Walter and Schmedders, Karl and Wilms, Ole},
  journal={The Journal of Finance},
  volume={73},
  number={3},
  pages={1061--1111},
  year={2018},
  publisher={Wiley Online Library}
}

@article{tauchen1986finite,
  title={Finite state markov-chain approximations to univariate and vector autoregressions},
  author={Tauchen, George},
  journal={Economics letters},
  volume={20},
  number={2},
  pages={177--181},
  year={1986},
  publisher={Elsevier}
}

@article{christiano1990solving,
  title={Solving the stochastic growth model by linear-quadratic approximation and by value-function iteration},
  author={Christiano, Lawrence J},
  journal={Journal of Business \& Economic Statistics},
  volume={8},
  number={1},
  pages={23--26},
  year={1990},
  publisher={Taylor \& Francis}
}

@article{backus2004exotic,
  title={Exotic preferences for macroeconomists},
  author={Backus, David K and Routledge, Bryan R and Zin, Stanley E},
  journal={NBER Macroeconomics Annual},
  volume={19},
  pages={319--390},
  year={2004},
  publisher={MIT Press}
}

@article{lopez2018risk,
  title={Risk-Adjusted Linearizations of Dynamic Equilibrium Models},
  author={Lopez, Pierlauro and Lopez-Salido, David and Vazquez-Grande, Francisco},
  year={2018}
}

@article{petrosky2018endogenous,
Author = {Petrosky-Nadeau, Nicolas and Zhang, Lu and Kuehn, Lars-Alexander},
Title = {Endogenous Disasters},
Journal = {American Economic Review},
Volume = {108},
Number = {8},
Year = {2018},
Month = {August},
Pages = {2212-45},
DOI = {10.1257/aer.20130025},
URL = {http://www.aeaweb.org/articles?id=10.1257/aer.20130025}}

@article{hansen2001robust,
  title={Robust control and model uncertainty},
  author={Hansen, LarsPeter and Sargent, Thomas J},
  journal={American Economic Review},
  volume={91},
  number={2},
  pages={60--66},
  year={2001}
}

@article{decker2016market,
  title={Market exposure and endogenous firm volatility over the business cycle},
  author={Decker, Ryan A and D'Erasmo, Pablo N and Moscoso Boedo, Hernan},
  journal={American Economic Journal: Macroeconomics},
  volume={8},
  number={1},
  pages={148--98},
  year={2016}
}

@article{mccallum1983non,
  title={On non-uniqueness in rational expectations models: An attempt at perspective},
  author={McCallum, Bennett T},
  journal={Journal of monetary Economics},
  volume={11},
  number={2},
  pages={139--168},
  year={1983},
  publisher={Elsevier}
}

@article{albuquerque2016valuation,
  title={Valuation risk and asset pricing},
  author={Albuquerque, Rui and Eichenbaum, Martin and Luo, Victor Xi and Rebelo, Sergio},
  journal={The Journal of Finance},
  volume={71},
  number={6},
  pages={2861--2904},
  year={2016},
  publisher={Wiley Online Library}
}

@article{chen2017external,
  title={External Habit in a Production Economy: A Model of Asset Prices and Consumption Volatility Risk},
  author={Chen, Andrew Y},
  journal={The Review of Financial Studies},
  year={2017}
}

@article{colacito2014bkk,
  title={Bkk the ez way},
  author={Colacito, Riccardo and Croce, Mariano and Ho, Steven and Howard, Philip},
  journal={Wor\&ing Paper, University of North Carolina},
  year={2014}
}

@article{basu2017uncertainty,
  title={Uncertainty shocks in a model of effective demand},
  author={Basu, Susanto and Bundick, Brent},
  journal={Econometrica},
  volume={85},
  number={3},
  pages={937--958},
  year={2017},
  publisher={Wiley Online Library}
}

@article{cochrane2017macro,
  title={Macro-finance},
  author={Cochrane, John H},
  journal={Review of Finance},
  volume={21},
  number={3},
  pages={945--985},
  year={2017},
  publisher={Oxford University Press}
}

@article{backus2014sources,
  title={Sources of entropy in representative agent models},
  author={Backus, David and Chernov, Mikhail and Zin, Stanley},
  journal={The Journal of Finance},
  volume={69},
  number={1},
  pages={51--99},
  year={2014},
  publisher={Wiley Online Library}
}

@article{croce2014investor,
  title={Investor information, long-run risk, and the term structure of equity},
  author={Croce, Mariano M and Lettau, Martin and Ludvigson, Sydney C},
  journal={The Review of Financial Studies},
  volume={28},
  number={3},
  pages={706--742},
  year={2014},
  publisher={Oxford University Press}
}

@article{klibanoff2005smooth,
  title={A smooth model of decision making under ambiguity},
  author={Klibanoff, Peter and Marinacci, Massimo and Mukerji, Sujoy},
  journal={Econometrica},
  volume={73},
  number={6},
  pages={1849--1892},
  year={2005},
  publisher={Wiley Online Library}
}

@article{ilut2014ambiguous,
  title={Ambiguous business cycles},
  author={Ilut, Cosmin L and Schneider, Martin},
  journal={The American Economic Review},
  volume={104},
  number={8},
  pages={2368--2399},
  year={2014},
  publisher={American Economic Association}
}

@article{brunnermeier2014macroeconomic,
  title={A macroeconomic model with a financial sector},
  author={Brunnermeier, Markus K and Sannikov, Yuliy},
  journal={The American Economic Review},
  volume={104},
  number={2},
  pages={379--421},
  year={2014},
  publisher={American Economic Association}
}

@article{kopecky2010finite,
  title={Finite state Markov-chain approximations to highly persistent processes},
  author={Kopecky, Karen A and Suen, Richard MH},
  journal={Review of Economic Dynamics},
  volume={13},
  number={3},
  pages={701--714},
  year={2010},
  publisher={Elsevier}
}

@article{yang2015intertemporal,
author = {Yang, Wei}, 
title = {Intertemporal Substitution and Equity Premium},
year = {2015}, 
doi = {10.1093/rof/rfv004}, 
URL = {http://rof.oxfordjournals.org/content/early/2015/03/21/rof.rfv004.abstract}, 
eprint = {http://rof.oxfordjournals.org/content/early/2015/03/21/rof.rfv004.full.pdf+html}, 
journal = {Review of Finance} 
}

@article{campbell1999force,
  title={By Force of Habit: A Consumption-Based Explanation of Aggregate Stock Market Behavior},
  author={Campbell, John Y and Cochrane, John H},
  journal={Journal of Political Economy},
  volume={107},
  number={2},
  pages={205--251},
  year={1999},
  publisher={JSTOR}
}

@article{ju2012ambiguity,
  title={Ambiguity, learning, and asset returns},
  author={Ju, Nengjiu and Miao, Jianjun},
  journal={Econometrica},
  volume={80},
  number={2},
  pages={559--591},
  year={2012},
  publisher={Wiley Online Library}
}

@article{barro1984time,
  title={Time-Separable Preferences and Intertemporal-Substitution Models of Business Cycles},
  author={Barro, Robert J and King, Robert G},
  journal={The Quarterly Journal of Economics},
  volume={99},
  number={4},
  pages={817--839},
  year={1984},
  publisher={Oxford University Press}
}

@article{dupor2014analytics,
  title={The analytics of technology news shocks},
  author={Dupor, Bill and Mehkari, M Saif},
  journal={Journal of Economic Theory},
  volume={153},
  pages={392--427},
  year={2014},
  publisher={Elsevier}
}

@article{malkhozov2014asset,
  title={Asset prices in affine real business cycle models},
  author={Malkhozov, Aytek},
  journal={Journal of Economic Dynamics and Control},
  volume={45},
  pages={180--193},
  year={2014},
  publisher={Elsevier}
}

@article{uhlig1999toolkit,
  title={A toolkit for analyzing nonlinear dynamic stochastic models easily},
	year = {1999},
	journal = {Computational Methods for the Study of Dynamic Economies},	
  author={Harald Uhlig},
  publisher={Oxford University Press, Oxford}
}

@article{backus2015risk,
  title={Risk and ambiguity in models of business cycles},
  author={Backus, David and Ferriere, Axelle and Zin, Stanley},
  journal={Journal of Monetary Economics},
  volume={69},
  pages={42--63},
  year={2015},
  publisher={Elsevier}
}

@article{uhlig2010easy,
  title={Easy ez in dsge},
  author={Uhlig, Harald},
  journal={Unpublished Manuscript, University of Chicago},
  year={2010}
}

@article{bansal2012empirical,
  title={An Empirical Evaluation of the Long-Run Risks Model for Asset Prices},
  author={Bansal, Ravi and Kiku, Dana and Yaron, Amir},
  journal={Critical Finance Review},
  volume={1},
  number={1},
  pages={183--221},
  year={2012},
  publisher={now publishers}
}

@article{schmitt2004solving,
  title={Solving dynamic general equilibrium models using a second-order approximation to the policy function},
  author={Schmitt-Grohe, Stephanie and Uribe, Mart{\i}n},
  journal={Journal of economic dynamics and control},
  volume={28},
  number={4},
  pages={755--775},
  year={2004},
  publisher={Elsevier}
}

@article{rudebusch2012bond,
  title={The bond premium in a dsge model with long-run real and nominal},
  author={Rudebusch, Glenn D and Swanson, Eric T},
  journal={American Economic Journal: Macroeconomics},
  volume={4},
  number={1},
  pages={105--143},
  year={2012},
  publisher={American Economic Association}
}

@article{lucas2003macroeconomic,
  title={Macroeconomic priorities},
  author={Lucas, Robert E},
  journal={American Economic Review},
  volume={93},
  number={1},
  pages={1--14},
  year={2003},
  publisher={Princeton, NJ: American Economic Association, 1911-}
}

@article{van2012term,
  title={The term structure of interest rates in a DSGE model with recursive preferences},
  author={Van Binsbergen, Jules H and Fern{\'a}ndez-Villaverde, Jes{\'u}s and Koijen, Ralph SJ and Rubio-Ramirez, Juan},
  journal={Journal of Monetary Economics},
  volume={59},
  number={7},
  pages={634--648},
  year={2012},
  publisher={Elsevier}
}

@article{caldara2012computing,
  title={Computing DSGE models with recursive preferences and stochastic volatility},
  author={Caldara, Dario and Fernandez-Villaverde, Jesus and Rubio-Ramirez, Juan F and Yao, Wen},
  journal={Review of Economic Dynamics},
  volume={15},
  number={2},
  pages={188--206},
  year={2012},
  publisher={Elsevier}
}

@article{belo2010production,
  title={Production-based measures of risk for asset pricing},
  author={Belo, Frederico},
  journal={Journal of Monetary Economics},
  volume={57},
  number={2},
  pages={146--163},
  year={2010},
  publisher={Elsevier}
}

@article{van2006learning,
  title={Learning asymmetries in real business cycles},
  author={Van Nieuwerburgh, Stijn and Veldkamp, Laura},
  journal={Journal of Monetary Economics},
  volume={53},
  number={4},
  pages={753--772},
  year={2006},
  publisher={Elsevier}
}

@article{bekaert2017asset,
  title={Asset return dynamics under habits and bad environment--good environment fundamentals},
  author={Bekaert, Geert and Engstrom, Eric},
  journal={Journal of Political Economy},
  volume={125},
  number={3},
  pages={713--760},
  year={2017},
  publisher={University of Chicago Press Chicago, IL}
}

@article{gourio2012disaster,
  title={Disaster Risk and Business Cycles},
  author={Gourio, Fran{\c{c}}ois},
  journal={The American Economic Review},
  volume={102},
  number={6},
  pages={2734--66},
  year={2012}
}

@article{gabaix2012variable,
  title={Variable rare disasters: An exactly solved framework for ten puzzles in macro-finance},
  author={Gabaix, Xavier},
  journal={The Quarterly Journal of Economics},
  volume={127},
  number={2},
  pages={645--700},
  year={2012},
  publisher={Oxford University Press}
}

@article{croce2014long,
title = "Long-run productivity risk: A new hope for production-based asset pricing? ",
journal = "Journal of Monetary Economics ",
volume = "66",
number = "0",
pages = "13 - 31",
year = "2014",
note = "",
issn = "0304-3932",
author = "Mariano Massimiliano Croce"
}

@article{abel1983optimal,
  title={Optimal investment under uncertainty},
  author={Abel, A.B.},
  journal={The American Economic Review},
  pages={228--233},
  year={1983},
  publisher={JSTOR}
}

@article{tallarini2000risk,
  title={Risk-sensitive real business cycles},
  author={Tallarini, Thomas D.},
  journal={Journal of Monetary Economics},
  volume={45},
  number={3},
  pages={507--532},
  year={2000},
  publisher={Elsevier}
}

@misc{rouwenhorst1995asset,
  title={Asset pricing implications of equilibrium business cycle models, chapter 10, 294 330},
  author={Rouwenhorst, K Geert},
  year={1995},
	journal={Frontiers of Business Cycle Research},
  publisher={Princeton University Press}
}

@article{hayashi1982tobin,
  title={Tobin's marginal q and average q: A neoclassical interpretation},
  author={Hayashi, F.},
  journal={Econometrica: Journal of the Econometric Society},
  pages={213--224},
  year={1982},
  publisher={JSTOR}
}

@article{epstein1989substitution,
  title={Substitution, risk aversion, and the temporal behavior of consumption and asset returns: A theoretical framework},
  author={Epstein, L.G. and Zin, S.E.},
  journal={Econometrica: Journal of the Econometric Society},
  pages={937--969},
  year={1989},
  publisher={JSTOR}
}

@article{wachter2013can,
  title={Can Time-Varying Risk of Rare Disasters Explain Aggregate Stock Market Volatility?},
  author={Wachter, Jessica A},
  journal={The Journal of Finance},
  volume={68},
  number={3},
  pages={987--1035},
  year={2013},
  publisher={Wiley Online Library}
}

@article{judd1992projection,
  title={Projection methods for solving aggregate growth models},
  author={Judd, K.L.},
  journal={Journal of Economic Theory},
  volume={58},
  number={2},
  pages={410--452},
  year={1992},
  publisher={Elsevier}
}

@misc{miranda2001applied,
  title={Applied computational economics},
  author={Miranda, M. and Fackler, P.},
  year={2001},
  publisher={MIT Press, Cambridge, MA}
}

@article{kaltenbrunner2010long,
  title={Long-run risk through consumption smoothing},
  author={Kaltenbrunner, G. and Lochstoer, L.A.},
  journal={Review of Financial Studies},
  volume={23},
  number={8},
  pages={3190--3224},
  year={2010},
  publisher={Soc Financial Studies}
}

@article{bloom2009impact,
  title={The Impact of Uncertainty Shocks},
  author={Bloom, N.},
  journal={Econometrica},
  volume={77},
  number={3},
  pages={623--685},
  year={2009},
  publisher={Wiley Online Library}
}

@article{bansal2005interpretable,
  title={Interpretable asset markets?},
  author={Bansal, Ravi and Khatchatrian, Varoujan and Yaron, Amir},
  journal={European Economic Review},
  volume={49},
  number={3},
  pages={531--560},
  year={2005},
  publisher={Elsevier}
}

@article{bansal2004risks,
  title={Risks for the long run: A potential resolution of asset pricing puzzles},
  author={Bansal, Ravi and Yaron, Amir},
  journal={The Journal of Finance},
  volume={59},
  number={4},
  pages={1481--1509},
  year={2004},
  publisher={Wiley Online Library}
}

@article{jermann1998asset,
  title={Asset pricing in production economies},
  author={Jermann, Urban J.},
  journal={Journal of Monetary Economics},
  volume={41},
  number={2},
  pages={257--275},
  year={1998},
  publisher={Elsevier}
}

@article{papanikolaou2011investment,
  title={Investment Shocks and Asset Prices},
  author={Papanikolaou, Dimitris},
  journal={Journal of Political Economy},
  volume={119},
  number={4},
  pages={639--685},
  year={2011},
  publisher={University of Chicago Press}
}

@article{guvenen2009parsimonious,
  title={{A parsimonious macroeconomic model for asset pricing}},
  author={Guvenen, Fatih},
  journal={Econometrica},
  volume={77},
  number={6},
  pages={1711--1750},
  issn={1468-0262},
  year={2009},
  publisher={Wiley Online Library}
}

\end{document}